%% file: kpartition.tex
\documentclass[11pt]{article}

\usepackage[utf8]{inputenc}
\usepackage[english]{babel}

\usepackage{amsmath,amssymb,amsfonts,amsthm}
\usepackage{mathabx}

\usepackage{url}

\usepackage{fullpage}

\usepackage{cleveref}
\newtheorem{theorem}{Theorem}
\newtheorem{lemma}{Lemma}
\newtheorem{corollary}{Corollary}
\theoremstyle{definition}

\newcommand{\calR}{\mathcal R}

\newcommand{\ceil}[1]{\left\lceil{#1}\right\rceil}
\newcommand{\floor}[1]{\left\lfloor{#1}\right\rfloor}
\newcommand{\cond}{\;|\;}
\newcommand{\eps}{\varepsilon}

\newcommand{\E}{\mathbf{E}}
\newcommand{\xor}{\oplus}
\newcommand{\bigxor}{\bigoplus}
\newcommand{\Prp}[1]{\Pr\!\left[{#1} \right]}
\newcommand{\Ep}[1]{\E\!\left[{#1} \right]}

\newcommand{\sm}{\,\setminus\,}
\newcommand{\abs}[1]{\left | #1 \right |}
\newcommand{\set}[1]{\left \{ #1 \right \}}
\newcommand{\oneToN}[1]{\set{1, \ldots, #1}}

\usepackage[usenames,dvipsnames,table]{xcolor}

\newcommand{\errterm}{O\!\left(|\Sigma|^{1-\floor{d/2}}\right)}
\newcommand{\terrterm}{\tilde{O}\!\left(|\Sigma|^{1-\floor{d/2}}\right)}

\newcommand\numberthis{\addtocounter{equation}{1}\tag{\theequation}}

\usepackage{graphicx}

\usepackage{booktabs}

\crefname{lemma}{Lemma}{Lemmas}
\Crefname{lemma}{Lemma}{Lemmas}
\crefname{theorem}{Theorem}{Theorems}
\Crefname{theorem}{Theorem}{Theorems}
\crefname{observation}{Observation}{Observations}
\Crefname{observation}{Observation}{Observations}
\crefname{definition}{Definition}{Definitions}
\Crefname{definition}{Definition}{Definitions}
\crefname{section}{Section}{Sections}
\Crefname{section}{Section}{Sections}
\crefname{figure}{Figure}{Figures}
\Crefname{figure}{Figure}{Figures}
\crefname{appendix}{Appendix}{Appendices}
\Crefname{appendix}{Appendix}{Appendices}

\usepackage{todonotes}
\newcommand{\sd}[1]{{#1}}
\newcommand{\mt}[1]{{#1}}

\newcommand*\samethanks[1][\value{footnote}]{\footnotemark[#1]}
\usepackage{authblk}

\title{Hashing for statistics over $k$-partitions}
\author{Søren Dahlgaard\thanks{Research partly supported by Mikkel Thorup's
    Advanced Grant DFF-0602-02499B from the Danish Council for Independent Research
    under the Sapere Aude research career programme.}}
\author{Mathias Bæk Tejs Knudsen\samethanks[1]\ \thanks{Research partly supported by the FNU
    project AlgoDisc - Discrete Mathematics, Algorithms, and Data Structures.}}
\author{Eva Rotenberg}
\author{Mikkel Thorup\samethanks[1]}
\affil{University of Copenhagen,\\
    \tt{\{soerend,knudsen,roden,mthorup\}@di.ku.dk}
}
\date{}

\begin{document}

\setcounter{page}{0}
\maketitle
\input{abstract}
\thispagestyle{empty}

\newpage
\setcounter{page}{1}
\input{introduction}
\input{minhash}

\input{bounddep_full}
\input{uniform}
\input{uniform_multip}
\input{kthmoment}

\newpage
\bibliographystyle{plain}
\bibliography{general}

\end{document}

%% file: abstract.tex
\begin{abstract}
In this paper we analyze a hash function for $k$-partitioning a set into bins,
obtaining strong concentration bounds for standard algorithms combining
statistics from each bin.

This generic method was originally introduced by
Flajolet and Martin~[FOCS'83] in order to save a factor $\Omega(k)$ of time per
element over $k$ independent samples when estimating the number of distinct
elements in a data stream. It was also used in the widely used HyperLogLog
algorithm of Flajolet et al.~[AOFA'97] and in large-scale machine learning by
Li et al.~[NIPS'12] for minwise estimation of set similarity.

The main issue of $k$-partition, is that the contents of different bins may be
highly correlated when using popular hash functions. This means that
methods of analyzing the marginal distribution for a single bin do not apply.
Here we show that a tabulation based hash function, mixed tabulation, does
yield strong concentration bounds on the most popular applications of
$k$-partitioning similar to those we would get using a truly random hash
function.
The analysis is very involved and implies several new results of independent
interest for both simple and double tabulation, e.g. a simple and
efficient construction for \sd{invertible bloom filters and} uniform hashing on
a given set.

\end{abstract}

%% file: introduction.tex
\section{Introduction}
\def\polylog{\textnormal{polylog}}
\def\tO{\tilde O}
\def\ol{\overline}
\def\ul{\underline}
\def\cR{\mathcal R}
\def\cM{\mathcal M}
\newcommand\mikcom[1]{\marginpar{Mikkel}\begin{quote}{\em {#1}}
\end{quote}}

A useful assumption in the design of randomized algorithms and data
structures is the free availability of fully random hash functions
which can be computed in unit time. Removing this unrealistic assumption is the
subject of a large body of work. To implement a hash-based algorithm,
a concrete hash function has to be chosen. The space, time, and random
choices made by this hash function affects the overall performance.
{\em The generic goal is therefore to provide efficient constructions of hash
  functions that for important randomized algorithms yield
  probabilistic guarantees similar to those obtained assuming fully
  random hashing. }

To fully appreciate the significance of this program, we note that many
randomized algorithms are very simple and popular in practice, but often
they are implemented with too simple hash functions without the necessary
guarantees. This may work very well in random tests, adding to their
popularity, but the real world is full of structured data
that could be bad for the hash function.
This was illustrated in \cite{thorup12kwise}
showing how simple common inputs made linear probing fail with popular
hash functions, explaining its perceived unreliability in practice. The
problems disappeared when sufficiently strong hash functions were used.

In this paper, we consider the generic approach where a hash function
is used to $k$-partition a set into bins. Statistics are computed on
each bin, and then all these statistics are combined so as to get good
concentration bounds. This approach was introduced by Flajolet and
Martin \cite{Flajolet85counting} under the name \emph{stochastic
averaging} to estimate the number of distinct elements in a data
stream. \sd{Today, a more popular estimator of this quantity
is the HyperLogLog counter, which is also based on $k$-partitioning
\cite{Flajolet07hyperloglog,Heule13hyperloglog}. These types of counters have found many applications, e.g., to estimate the neighbourhood function of a
graph with all-distance sketches \cite{boldi11hyperanf,Cohen14ads}.}

\sd{Later it was considered by Li et al.
\cite{li12oneperm,Shrivastava14oneperm,Shrivastava14densify} in the classic
minwise hashing framework of Broder et al. for the very different application
of set similarity estimation
\cite{broder97minwise,broder98minwise,broder97onthe}. To our knowledge we are
the first to address such statistics over a $k$-partitioning with practical hash
functions.}

\sd{We will use the example of MinHash for frequency estimation as a running
example throughout the paper:}
suppose we have a fully random hash function
applied to a set $X$ of red and blue balls.  We want to estimate the
fraction $f$ of red balls. The idea of the MinHash algorithm is to
sample the ball with the smallest hash value. With a fully-random hash
function, this is a uniformly random sample from $X$, and it is red with probability $f$.
For better concentration, we may use {\em $k$ independent repetitions}: we repeat the experiment
$k$ times with $k$ independent hash functions. This yields a multiset $S$
of $k$ samples with replacement from $X$.  The fraction of red balls
in $S$ concentrates around $f$ and the error probability falls
exponentially in $k$.

Consider now the alternative experiment based on {\em $k$-partitioning}:
we use a single hash function, where the first $\ceil{\lg k}$ bits of the hash value partition $X$ into $k$ bins, and then the remaining bits are used as a
{\em local hash value} within the bin. We pick the ball with the
smallest (local) hash value in each bin. This is a sample $S$ from $X$
without replacement, and again, the fraction of red balls \sd{in the non-empty
bins} is
concentrated around $f$ with exponential concentration bounds. We note
that there are some differences. We do get the advantage that the
samples are without replacement, which means better concentration. On
the other hand, we may end up with fewer samples if some bins are
empty.

The big difference between the two schemes is that the second one runs
$\Omega(k)$ times faster. In the first experiment, each ball
participated in $k$ independent experiments, but in the second one with
$k$-partitioning, each ball picks its bin, and then only
participates in the local experiment for that bin. Thus, essentially, we get
$k$ experiments for the price of one. Handling each ball, or key, in constant
time is important in applications of high volume streams.


In this paper, we present the first realistic hash function for
$k$-partitioning in these application. Thus we will get concentration bounds
similar to those obtained with fully random hashing for the following
algorithms:
\begin{itemize}
\item Frequency/similarity  estimation as in our running example and
  as it is used for the machine learning
  in~\cite{li12oneperm,Shrivastava14oneperm,Shrivastava14densify}.
\item Estimating distinct elements as in
  \cite{Flajolet85counting,Flajolet07hyperloglog}.
\end{itemize}
\sd{
Other technical developments include simpler hash functions for invertible
Bloom filters, uniform hashing, and constant moment bounds.

For completeness we mention that the count sketch data structure of Charikar et
al. \cite{Charikar02countsketch} is also based on $k$-partitioning. However, for
count sketches we can never hope for the kind of strong concentration bounds
pursued in this paper as they are prevented by the presence of large weight
items. The analysis in \cite{Charikar02countsketch} is just based on variance for which $2$-independent
hashing suffices. Strong concentration bounds are instead obtained by
independent repetitions.}

\subsection{Applications in linear machine learning}
As mentioned, our running example with red and blue balls is
mathematically equivalent to the classic application of minwise
hashing to estimate the Jaccard similarity $J(X,Y) = |X\cap Y|/|X\cup Y|$
between two sets $X$ and $Y$. This method was originally introduced by
Broder et al.~\cite{broder97minwise,broder98minwise,broder97onthe} for
the AltaVista search engine. The red balls correspond to the
intersection of $X$ and $Y$ and the blue balls correspond to the symmetric
difference. The MinHash \sd{estimator is the indicator variable of whether the
ball with the smallest hash value over both sets belongs to the intersection of
the two sets.} To determine this we store the smallest hash value from each set
as a \emph{sketch} and check if it is the same.
\sd{In order to reduce the variance one uses $k$
independent hash functions, known as $k\times$minwise. This method was later
revisited by Li et al.~\cite{Li10bbit,Li10bbit2,li11minhash}. By only using the
$b$ least significant bits of each hash value (for some small constant $b$),
they were able to create efficient linear sketches, encoding set-similarity as
inner products of sketch vectors, for use in large-scale
learning. However, applying $k$ hash functions to each element increases the
sketch creation time by roughly a factor of $k$.

It should be noted that
Bachrach and Porat \cite{Bachrach13sketching} have
suggested a more efficient way of maintaining $k$ Minhash values with
$k$ different hash functions. They use $k$ different polynomial hash
functions that are related, yet pairwise independent, so that they can
systematically maintain the Minhash for all $k$ polynomials in $O(\log
k)$ time per key assuming constant degree polynomials. There are two issues
with this approach: It is specialized to work with polynomials and Minhash is
known to have constant bias with constant degree polynomials \cite{patrascu10kwise-lb}, and this bias does not
decay with independent repetitions. Also, because the experiments are only
pairwise independent, the concentration is only limited by Chebyshev's
inequality.

An alternative to $k\times$minwise when estimating set similarity with
minwise sketches is bottom-$k$. In bottom-$k$ we use one hash function
and maintain the sketch as the keys with the $k$ smallest hash
values. This method can be viewed as sampling without replacement. Bottom-$k$
has been proved to work with simple hash functions both for counting
distinct elements \cite{bar-yossef02distinct} and for set similarity
\cite{thorup13bottomk}. However, it needs a priority queue to
maintain the $k$ smallest hash values and this leads to a non-constant
worst-case time per element, which may be a problem in real-time
processing of high volume data streams. A major problem in our context is that
we are not able to encode set-similarity as an inner product of two sketch
vectors. This is because the elements lose their ``alignment'' -- that
is, the key with the smallest hash value in one set might have the 10th
smallest hash value in another set.}

%

Getting down to constant time per element via $k$-partitioning was suggested by Li et
al.~\cite{li12oneperm,Shrivastava14oneperm,Shrivastava14densify}. They use
$k$-partitioned MinHash to quickly create small sketches of very
high-dimensional indicator vectors. \sd{Each sketch consists of $k$ hash values
corresponding to the hash value of each bin.} The sketches are then converted
into sketch vectors that code similarity as inner products. Finally the sketch
vectors are fed to a linear SVM for classifying massive data sets. \sd{The
sketches are illustrated in \Cref{fig:kpart_sketch}.
\begin{figure}[htbp]
    \centering
    \includegraphics[width=.3\textwidth]{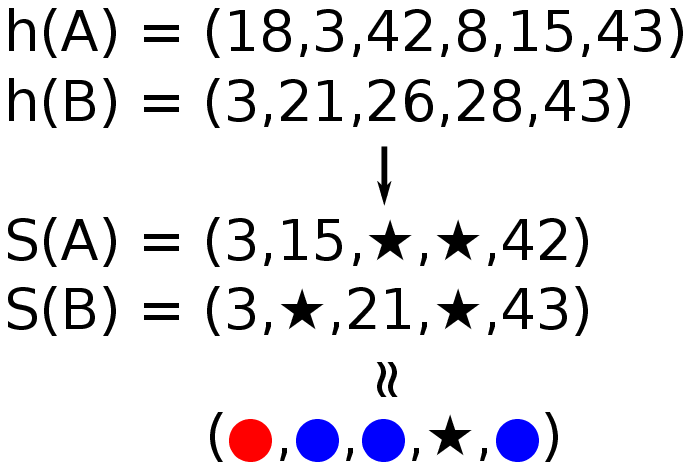}
    \caption{Example of $k$-partitioned sketches for two sets $A$ and $B$. The
        hash values are from the set $\{0,\ldots,49\}$ and $k=5$. The sketches
        $S(A)$ and $S(B)$ show the hash values for each bin and the $\star$
        symbol denotes an empty bin. The corresponding interpretation as red
        and blue balls is shown below with a red ball belonging to the
    intersection and blue ball to the symmetric difference. Here $k^\star = 4$.}
    \label{fig:kpart_sketch}
\end{figure}
Li et al.
also apply this approach to near neighbour search using locality sensitive
hashing as introduced in \cite{indyk98ann} (see
also~\cite{andoni08ann,Andoni14lsh}). When viewing the problems as red and blue
balls, the canonical \mt{unbiased estimator uses the number $k^\star$ of
non-empty bins, estimating $f$ as:
\begin{equation}\label{eq:estimator}
    \frac{\text{\# of red bins}}{k^\star}
\end{equation}}
A major complication of this estimator is that we do not know in advance, which
bins are jointly empty for two sketches (as illustrated in
\Cref{fig:kpart_sketch}). \mt{This means that there is no natural
way of computing the estimator as an inner product of the two sketches.}
Shrivastava and Li
\cite{Shrivastava14oneperm,Shrivastava14densify} suggest methods for dealing
with this by assigning empty bins a value by copying from non-empty
bins in different ways giving provably good bounds. It is important to note that
when all bins are non-empty, all the estimators considered in
\cite{li12oneperm,Shrivastava14oneperm,Shrivastava14densify} are identical to
the estimator in \eqref{eq:estimator} as $k^\star = k$ in this case.}

We note that for our running example with red and blue balls, it would
suffice to generate hash values on the fly, e.g., using a
pseudo-random number generator, but in the real application of set
similarity, it is crucial that when sets get processed separately,
the elements from the intersection get the same hash value. Likewise,
when we want to estimate distinct elements or do count sketches, it is
crucial that the same element always gets the same hash value.

\subsection{Technical challenge}\label{sec:challenge}
Using a hash function function to $k$-partition $n$ keys is often cast
as using it to throw $n$ balls into $k$ bins, which is an important
topic in randomized algorithms \cite[Chapter 3]{motwani95book}
\cite[Chapter 5]{mitzenmacher05book}.
However, when it comes to implementation via realistic hash functions,
the focus is often only on the marginal distribution in a single bin. For
example, with $k=n$, w.h.p., a given bin has load $O(\log n/\log\log
n)$, hence by a union bound, w.h.p., the maximum load is $O(\log
n/\log\log n)$. The high probability bound on the load of a given bin
follows with an $O(\log n/\log\log n)$-independent hash function
\cite{wegman81kwise}, but can also be obtained in other ways
\cite{CRSW11,patrascu12charhash}.

However, the whole point in using a $k$-partition is that we want to
average statistics over all $k$ bins hoping to get strong concentration
bounds, but this requires that the statistics from the
$k$ bins are not too correlated (even with full randomness, there is
always some correlation since the partitioning corresponds to sampling
without replacement, but this generally works in our favor).

To be more concrete, consider our running example with red and blue
balls where Minhash is used to pick a random ball from each bin. The
frequency $f$ of red balls is estimated as the frequency of red balls
in the sample.  Using $O(\log k)$ independent hashing, we can make sure that the bias in the sample
from any given bin is $1/k$ \cite{indyk01minwise}.  However, for concentration bounds on the average, we
have to worry about two types of correlations between statistics of
different bins. The first ``local'' correlation issue is if the local
hashing within different bins is too correlated. This issue
could conceivably be circumvented using one hash function for the
$k$-partitioning itself, and then have an independent local hash
function for each bin. The other ``global'' correlation issue is
for the overall $k$-partitioning distribution between
bins. It could be that if we get a lot of red balls in one bin, then
this would be part of a general clustering of the red balls on a few
bins (examples showing how such systematic clustering can happen with
simple hash functions are given in \cite{patrascu10kwise-lb}). This
clustering would disfavor the red balls in the overall average even if the
sampling in each bin was uniform and independent. This is an issue
of non-linearity, e.g., if there are already
more red than blue balls in a bin, then doubling their number only
increases their frequency by at most $3/2$. \sd{As mentioned earlier we are not
aware of any previous work addressing these issues with a less than fully
random hash function, but} for our running
example it appears that a $O(k \log k)$-independent hash function
will take care of both correlation issues (we will not prove this as we are
going to present an even better solution).

\paragraph{Resource consumption} We are now going to consider
the resource consumption by the different hashing schemes discussed
above. The schemes discussed are summarized in
\Cref{tab:kpart}.

\begin{table*}[htbp]
    \centering
    \begin{tabular}{l | c c}
        \toprule
        \textbf{Technique} & \textbf{Evaluation time} & \textbf{Space (words)} \\
        \midrule
        Fully random hashing & $O(1)$ & $u=n^{O(1)}$ \\
        Fully random on $n$ keys whp.~\cite{PP08}  & $O(1)$ & $(1+o(1))n$ \\
        $\tO(k)$-independence \cite{Christiani15indep} & $O(1)$ &
        $k u^\eps$ \\
        Mixed tabulation (this paper) & $O(1)$ & $\tO(k)+\sd{u^\eps}$ \\
        \bottomrule
    \end{tabular}
    \caption{Resources of hashing techniques. \sd{Here, $\eps$ may be chosen as an
    arbitrarily small positive constant.}}
    \label{tab:kpart}
\end{table*}

First we assume that the key universe is of size polynomial in
the number $n$ of keys. If not, we first do a standard universe
reduction, applying a universal hash function \cite{carter77universal}
into an intermediate universe of size $u = n^{O(1)}$, expecting no
collisions. We could now, in principle, have a fully random hash function
over $[u]$.

We can get down to linear space using the construction of Pagh and
Pagh (PP) \cite{PP08}. Their hash function uses $O(n)$ words and is, w.h.p.,
fully random on any given set of $n$ keys. However, using $O(n)$
space is still prohibitive in most applications as the main motivation of
$k$-partitioning is exactly to create an estimator of size $k$ when $n$ is
so big that we cannot store the set. Additionally, we may not know $n$ in advance.

As indicated above, it appears that $\Theta(k\log k)$-independent
hashing suffices for MinHash. For this we can use the recent construction of
Christiani et al.~\cite{Christiani15indep}. 
\sd{Their construction} gets $\Theta(k\log k)$-independence, w.h.p., in
$O(1)$ time using space \sd{$k u^\eps$ for an arbitrarily small constant
$\eps$ affecting the evaluation time}. Interestingly, we
use the same space if we want a $\Theta(\log k)$-independent hash function for
each of the $k$ bins. The construction of Thorup \cite{thorup13doubletab}
gives independence $\sd{u^\eps}\gg \log k$ in $O(1)$ time using
\sd{$u^\eps$} space. A lower bound of Siegel \cite{siegel04hash} shows
that we cannot hope to improve the space in either case if we want
fast hashing. More precisely, if we want $q$-independence in time
$t<q$, we need space at least $q(u/q)^{1/t}$. Space $k u^{\Omega(1)}$
thus appears to be the best we can hope for with these independence
based approaches.

\subsection{$k$-partitions via mixed tabulation}
In this paper we present and analyze a hash function, {\em mixed tabulation},
that for all the $k$-partitioning algorithms discussed above, w.h.p., gets
concentration similar to that with fully random hash functions. The
hashing is done in $O(1)$ time and $\tO(k) + \sd{u^\eps}$ space.
If, say, $k=u^{\Omega(1)}$, this means that we hash in constant
time using space near-linear in the number of counters.
This is the first proposals of a hash function for statistics over
$k$-partitions that has good theoretical probabilistic
properties, yet does not significantly increase the amount of
resources used by these popular algorithms.
The hash function we suggest for $k$-partitioning, \emph{mixed tabulation},
is an extension of simple tabulation hashing.

\paragraph{Simple tabulation} Simple tabulation
hashing dates back to Zobrist \cite{zobrist70hashing}. \sd{The hash family
takes an integer parameter $c>1$, and} we view a key $x\in [u] =
\{0,\ldots,u-1\}$ as a vector of $c$ characters $x_0,\ldots,x_{c-1}\in \Sigma
= [u^{1/c}]$. The hash values are bit strings of some length $r$.
For each character position $i$, we initialize a fully random table
$T_i$ of size $|\Sigma|$ with values from $\calR = [2^r]$. The hash value of a key $x$ is calculated as
\[
    h(x) = T_0[x_0]\xor\cdots\xor T_{c-1}[x_{c-1}]\ .
\]
Simple tabulation thus takes time $O(c)$ and space $O(c u^{1/c})$.
In our context we assume that $c$ is a constant and that the character tables
fit in fast cache \sd{(eg. for $64$-bit keys we may pick $c=4$ and have
$16$-bit characters. The tables $T_i$ then take up $2^{16}$ words)}.
Justifying this assumption, recall that
with universe reduction, we  can assume that the universe is
of size $u = n^{O(1)}$. Now, for any desired constant $\eps>0$, we can pick $c=O(1)$ such that
$\Sigma = u^{1/c} \le n^\eps$. We refer to the lookups $T_i[x_i]$ as
\emph{character lookups} to emphasize that we expect them to be much faster
than a general lookups in memory. P\v{a}tra\c{s}cu and Thorup
\cite{patrascu12charhash} found simple tabulation to be 3 times faster
than evaluating a degree-2 polynomial over a prime field for the same
key domain.

P\v{a}tra\c{s}cu and Thorup \cite{patrascu12charhash} analyzed simple
tabulation assuming $c=O(1)$, showing that it works very well
for common applications of hash function such as linear probing, cuckoo hashing
and minwise independence. Note, however, that $O(\log n)$ independence was
known to suffice for all these applications. We also
note that simple tabulation fails to give good
concentration for $k$-partitions: Consider the set $R = [2]\times [m/2]$ of
$m$ red balls and let $B$ be some random set of blue balls. In this case the
red balls hash into the same buckets in pairs with probability $1/k$, which
will skew the estimate by a factor of $2$ if, for instance, $|R|$ is relatively
small.

\paragraph{Mixed tabulation}
To handle $k$-partitions, we here propose and analyze a mix between simple tabulation defined above and
the double tabulation scheme of \cite{thorup13doubletab}. \sd{In addition to
$c$,} mixed tabulation takes as a parameter an integer $d \ge 1$. We derive
$d$ extra characters using one simple tabulation function and compose
these with the original key before applying an extra round of simple
tabulation. Mathematically, we use two simple tabulation hash functions $h_1 :
\Sigma^c\to \Sigma^d$ and $h_2 : \Sigma^{d+c}\to \calR$ and define the hash function
to be $h(x)\mapsto h_2(x\cdot h_1(x))$, where $\cdot$ denotes concatenation
of characters.
We call $x \cdot h_1(x)$ the \emph{derived key} and denote this by
$h_1^\star(x)$. Our mixed tabulation scheme is very similar to Thorup's double
tabulation \cite{thorup13doubletab} and we shall return to the relation in
\Cref{sec:other}.
We note that we can implement this using just $c+d$ lookups
if we instead store simple tabulation functions
$h_{1,2}:\Sigma^c\to \Sigma^d\times\calR$ and $h_2':\Sigma^d\to \calR$,
computing $h(x)$ by $(v_1,v_2)=h_{1,2}(x);\ h(x)=v_1\xor h_2'(v_2)$.
This efficient implementation is similar to that of twisted tabulation
\cite{PT13:twist}, and is equivalent to the previous definition.
In our applications, we think of \sd{$c$ and} $d$ as a small constants, e.g.
\sd{$c=4$ and} $d=4$.
We note that we need not choose $\Sigma$ such that $|\Sigma|^c = u$.
Instead, we may pick $|\Sigma|\geq u^{1/c}$ to be any power of two.
A key $x$ is divided into $c$ characters $x_i$ of $b=\ceil{\lg u^{1/c}}$ or
$b-1$ bits,
so $x_i\in [2^b]\subseteq\Sigma$. This gives us the freedom to
use $c$ such that $u^{1/c}$ is not a power of two, but it
also allows us to work with $|\Sigma|\gg u^{1/c}$, which in effect
means that the derived characters are picked from a larger domain than
the original characters. Then mixed tabulation uses $O(c+d)$ time and $O(c u^{1/c}+d|\Sigma|)$ space.
For a good balance, we will always pick $c$ and $|\Sigma|$ such
that $u^{1/c}\leq |\Sigma|\leq u^{1/(c-1)}$. In all our applications we have
$c=O(1),d=O(1)$\mt{, which implies that the evaluation time is constant and
that the space used is $\Theta(|\Sigma|)$.}


\paragraph{Mixed tabulation in MinHash with $k$-partitioning}
We will now analyze MinHash with $k$-partitioning using mixed tabulation
as a hash function, showing that we get concentration bounds similar
to those obtained with fully-random hashing. The analysis is based on
two complimentary theorems. The first theorem states that for sets
of size nearly up to $|\Sigma|$, mixed tabulation is fully random
with high probability.
\begin{theorem}\label{thm:uniform_simple}
    Let $h$ be a mixed tabulation hash function with parameter $d$ and let
    $X\subseteq [u]$ be any input set. If $|X|\le |\Sigma|/(1+\Omega(1))$ then
    the keys of $X$ hash independently with probability $1 -
    O(|\Sigma|^{1-\floor{d/2}})$.
\end{theorem}
The second theorem will be used to analyze the performance for larger
sets. It is specific to MinHash with $k$-partitioning, stating, w.h.p.,
that mixed tabulation hashing performs as well as fully random hashing
with slight changes to the number of balls:
\begin{theorem}\label{thm:minwise}
    Consider a set of $n_R$ red balls and $n_B$ blue balls \sd{with
    $n_R + n_B > |\Sigma|/2$}. Let $f=n_R/(n_R+n_B)$ be the fraction of red
    balls which we wish to estimate.

    Let $X^{\mathcal M}$ be the estimator of $f$ from \eqref{eq:estimator} that
    we get using MinHash
    with $k$-partitioning using mixed tabulation hashing with $d$ derived
    characters,
    where $k \le |\Sigma|/(4d\log|\Sigma|)$.

    Let $\ol X^{\cR}$ be the same estimator in the alternative
    experiment where we use fully random hashing but with $\lfloor
    n_R(1+\eps)\rfloor$ red balls and $\lceil n_B(1-\eps)\rceil$ blue
    balls where
    $\eps=O\!\left(\sqrt{\frac{\log|\Sigma|(\log\log|\Sigma|)^2}{|\Sigma|}}\right)$.  Then
    \[
        \Prp{X^{\cM}\ge (1+\delta)f}\leq \Prp{\ol X^{\cR}\geq
(1+\delta)f}+\terrterm\textnormal.
    \]
    Likewise, for a lower bound, let $\ul X^{\mathcal R}$  be the
    estimator in the experiment using fully random hashing
    but with $\lceil n_R(1-\eps)\rceil$ red balls and $\lfloor
    n_B(1+\eps)\rfloor$ blue balls.  Then
    \[
        \Prp{X^{\cM}\le (1-\delta)f}\leq \Prp{\ul X^{\cR}\leq
(1-\delta)f}+\terrterm.
    \]
\end{theorem}

To apply the above theorems, we pick our parameters $k$ and $\Sigma$ such that
\begin{equation}\label{eq:k-Sigma}
    k \le
    \min\set{\frac{\abs{\Sigma}}{\log\abs{\Sigma}(\log\log|\Sigma|)^2},
             \frac{\abs{\Sigma}}{4d \log \abs{\Sigma}}}
\end{equation}
Recall that we have the additional constraint that $|\Sigma|\geq u^{1/c}$ for
some $c=O(1)$. Thus \eqref{eq:k-Sigma} is only relevant if want to partition into
$k=u^{\Omega(1)}$ bins. It forces us to use space $\Theta(|\Sigma|)=
\Omega(k\log k(\log\log k)^2)$.

With this setting of parameters, we run MinHash with $k$-partitioning over
a given input. Let $n_R$ and $n_B$ be the number of red and blue balls,
respectively. Our analysis will hold no matter which of the
estimators from \cite{li12oneperm,Shrivastava14oneperm,Shrivastava14densify}
we apply.

If $n_R+n_B\leq |\Sigma|/2$, we refer to
\Cref{thm:uniform_simple}. It implies that no matter which of
estimators from \cite{li12oneperm,Shrivastava14oneperm,Shrivastava14densify}
we apply, we can refer directly to the analysis done in \cite{li12oneperm,Shrivastava14oneperm,Shrivastava14densify} assuming fully random hashing.
All we have to do is to add an extra error probability of
$O(|\Sigma|^{1-\floor{d/2}})$.

Assume now that $n_R+n_B\geq |\Sigma|/2$. First we note that all bins
are non-empty w.h.p.
To see this, we only consider the first
$|\Sigma|/2\geq 2d k \log |\Sigma|$ balls. By
\Cref{thm:uniform_simple},
they hash fully randomly with probability $1-O(|\Sigma|^{1-\floor{d/2}})$,
and if so, the probability that some bin is empty is bounded by $k(1-1/k)^{2d k \log
  |\Sigma|}<k/|\Sigma|^{2d}$. Thus, all bins are non-empty with
probability $1-O(|\Sigma|^{1-\floor{d/2}})$.

Assuming that all bins are non-empty, all the estimators from
\cite{li12oneperm,Shrivastava14oneperm,Shrivastava14densify} are
identical to \eqref{eq:estimator}. This means that \Cref{thm:minwise}
applies no matter which of the estimators we use since the error
probability $\terrterm$ absorbs the probability that some bin is
empty. In addition, the first bound in \eqref{eq:k-Sigma} implies that
$\eps=O(1/\sqrt k)$ (which is reduced to $o(1/\sqrt k)$ if
$\Sigma=\omega(k\log k (\log\log k)^2))$. In principle this completes the
description of how close mixed tabulation brings us in performance to
fully random hashing.

To appreciate the impact of $\eps$, we first consider what guarantees
we can give with fully random hashing. We are still assuming
$n_R+n_B\geq |\Sigma|/2$ where $|\Sigma|\geq 4dk\log|\Sigma|$ as
implied by \eqref{eq:k-Sigma}, so the probability of an
empty bin is bounded by $k(1-1/k)^{|\Sigma|/2}<|\Sigma|^{1-2d}$. Assume
that all bins are non-emtpy, and let $f=n_R/(n_R+n_B)$ be
the fraction of red balls.  Then our estimator $X^{\cR}$ of $f$ is the
fraction of red balls among $k$ samples without replacement. In
expectation we get $fk$ red balls. For $\delta\leq 1$, the probability
that the number of red balls deviates by more than $\delta fk$ from
$fk$ is $2\exp(\Omega(\delta^2fk))$. This follows from a standard
application of Chernoff bounds without replacement
\cite{serfling74replacement}.
The probability of a relative error
$\abs{X^{\cR} - f}/f \ge t/\sqrt{fk}$ is thus bounded by $2e^{-\Omega(t^2)}$ for any
$t \le \sqrt{fk}$.

The point now is that $\eps=O(1/\sqrt k)=O(1/\sqrt{fk})$. In the fully random experiments
in \Cref{thm:minwise}, we replace
$n_R$ by $n_R'=(1\pm\eps)n_R$ and $n_B$ with $n_B'=(1\pm\eps)n_B$. Then
$X^{\cR}$ estimates $f'=n_R'/(n_R'+n_B')=(1\pm\eps)f$, so we have
$\Pr[\abs{X^{\cR} - f'}/f' \ge t/\sqrt{f'k}]\leq 2e^{-\Omega(t^2)}$.
However, since $\eps=O(1/\sqrt{k})$,
this implies $\Pr[\abs{X^{\cR} - f}/f\ge t/\sqrt{fk}]\leq 2e^{-\Omega(t^2)}$ for
any $t\leq \sqrt{fk}$.
The only difference is that $\Omega$ hides a smaller constant. Including
the probability of getting an empty bin, we get
$\Pr[\abs{X^{\cR} - f}\ge t\sqrt{f/k}]\leq 2e^{-\Omega(t^2)}+|\Sigma|^{1-2d}$ for
any $t\leq \sqrt{fk}$. Hence, by \Cref{thm:minwise},
$\Pr[\abs{X^{\cM} - f}\ge t\sqrt{f/k}]\leq 2e^{-\Omega(t^2)}+\terrterm$ for
any $t\leq \sqrt{fk}$.

Now if $n_B\leq n_R$ and $f\geq 1/2$, it gives better concentration bounds
to consider the  symmetric estimator $X^{\cM}_B=1-X^{\cM}$ for the fraction $f_B=1-f\leq f$ of blue balls. The analysis from above shows that
$\Pr[\abs{X^{\cM}_B - f_B}\ge t\sqrt{f_B/k}]\leq 2e^{-\Omega(t^2)}+\terrterm$ for
any $t\leq \sqrt{f_Bk}$. Here $\abs{X^{\cM}_B - f_B}=\abs{X^{\cM} - f}$,
so we conclude that
$\Pr[\abs{X^{\cM} - f}\ge t\sqrt{\min\{f,1-f\}/k}]\leq 2e^{-\Omega(t^2)}+\terrterm$ for
any $t\leq \sqrt{\min\{f,1-f\}/k}$. Thus we have proved:

\begin{corollary}
    \label{cor:minwise}
    We consider MinHash with $k$-partitioning using mixed tabulation
    with alphabet $\Sigma$ and $c,d=O(1)$, and where $k$ satisfies
    \eqref{eq:k-Sigma}.
    Consider a set of $n_R$ and $n_B$ red and blue balls, respectively, where
    $n_R+n_B > |\Sigma|/2$.
    Let $f=n_R/(n_R+n_B)$ be the fraction of red balls that we wish
    to estimate. Let $X^{\cM}$ be the estimator of $f$ we get
    from our MinHash with $k$-partitioning using mixed tabulation.
    The estimator may be that in \eqref{eq:estimator}, or
    any of the estimators from
    \cite{li12oneperm,Shrivastava14oneperm,Shrivastava14densify}.
    Then for every $0 \le t \le \sqrt{\min\{f,1-f\}k}$,
    \[
        \Prp{\abs{X^{\cM} - f} \ge t\sqrt{\min\{f,1-f\}/k}} \le 2e^{-\Omega(t^2)} +
        \terrterm)\ .
    \]
\end{corollary}
The significance of having errors in terms of $1-f$ is when the
fraction of red balls represent similarity as discussed earlier. This
gives us much better bounds for the estimation of very similar sets.

The important point above is not so much the exact bounds we get in
\Cref{cor:minwise}, but rather the way we translate bounds
with fully random hashing  to the case of mixed tabulation.


\paragraph{Mixed tabulation in distinct counting with $k$-partitioning}

We can also show that distinct counting with $k$-partitioning using
mixed tabulation as a hash function gives concentration bounds similar
to those obtained with fully-random hashing. With less than
$|\Sigma|/2$ balls, we just apply \Cref{thm:uniform_simple}, stating that mixed tabulation is fully random with high probability.  With more balls, we use the following analogue to Theorem \ref{thm:minwise}:

\begin{theorem}\label{thm:distinct}
    Consider a set of \sd{$n > |\Sigma|/2$} balls. Let $X^{\mathcal M}$ be the
    estimator of $n$ using either stochastic averaging
    \cite{Flajolet85counting} or HyperLogLog \cite{Flajolet07hyperloglog} over
    a $k$-partition with mixed tabulation hashing where
    $k\le|\Sigma|/(4d\log|\Sigma|)$. Let $\ol X^{\cR}$ be the same estimator in
    the alternative experiment where we use fully random hashing but with
    $\lfloor n(1+\eps)\rfloor$ balls where
    $\eps=O\!\left(\sqrt{\frac{\log|\Sigma|(\log\log|\Sigma|)^2}{|\Sigma|}}\right)$.  Then
   \[\Prp{X^{\cM}\ge (1+\delta)n}\leq \Prp{\ol X^{\cR}\geq
(1+\delta)n}+\terrterm\textnormal,\]
    Likewise, for a lower bound, let $\ul X^{\mathcal R}$  be the
    estimator in the experiment using fully random hashing
    but with $\lceil n(1-\eps)\rceil$ balls. Then
    \[\Prp{X^{\cM}\le (1-\delta)n}\leq \Prp{\ul X^{\cR}\leq
(1-\delta)n}+\terrterm.\]
\end{theorem}

Conceptually, the proof of \Cref{thm:distinct} is much
simpler than that of \Cref{thm:minwise} since there are no colors. However,
the estimators are harder to describe, leading to a more messy formal proof,
which we do not have room for in this conference paper.

\subsection{Techniques and other results}\label{sec:other}
Our analysis of mixed tabulation gives many new insights into both
simple and double tabulation. To prove \cref{thm:minwise} and
\cref{thm:distinct}, we will show \sd{a generalization of
\Cref{thm:uniform_simple} proving} that mixed tabulation behaves like a truly
random hash function on fairly large sets with high probability, even when some
of the output bits of the hash function are known. The exact statement is as
follows.
\begin{theorem}\label{thm:uniform2}
    Let $h = h_2\circ h_1^\star$ be a mixed tabulation hash function. Let
    $X\subseteq [u]$ be any input set. Let $p_1,\ldots,p_b$ be any $b$ bit
        positions, $v_1,\ldots,v_b\in\{0,1\}$ be desired bit values
        and let $Y$ be the set of keys $x\in X$ where $h(x)_{p_i} = v_i$ for
    all $i$.
    If $\Ep{|Y|} = |X|\cdot2^{-b}\le
    |\Sigma|/(1+\Omega(1))$, then the remaining bits of the hash values in $Y$
    are completely independent with probability $1 -
    O(|\Sigma|^{1-\floor{d/2}})$.
\end{theorem}
In connection with our $k$-partition applications, the specified
output bits will be used to select a small set of keys that
are critical to the final statistics, and for which we have fully
random hashing on the remaining bits.

In order to prove \cref{thm:uniform2} we develop a number of
structural lemmas in \cref{sec:bound_dep} relating to key dependencies
in simple tabulation. These lemmas provides a basis for showing
some interesting results for simple tabulation and double
tabulation, which we also include in this paper. These results are
briefly described below.

\paragraph{Double tabulation and uniform hashing} In double tabulation \cite{thorup13doubletab}, we
compose two independent simple tabulation functions $h_1:\Sigma^c\to\Sigma^d$ and
$h_2:\Sigma^d\to\calR$ defining $h:\Sigma^c\to\calR$ as
$h(x)=h_2(h_1(x))$. We note that with the same values for $c$ and $d$, double tabulation
is a strict simplification of mixed tabulation in that $h_2$ is only applied
to $h_1(x)$ instead of to $x\cdot h_1(x)$. The advantage of
mixed tabulation is that we know that the ``derived'' keys $x\cdot h_1(x)$
are distinct, and this is crucial to our analysis of $k$-partitioning.
However, if all we want is uniformity over a given set, then we
show that the statement of \cref{thm:uniform_simple} also holds for
double tabulation.
\begin{theorem}\label{thm:double-uniform}
    Given an arbitrary set $S\subseteq [u]$ of size $|\Sigma|/(1+\Omega(1))$,
with probability
$1-O(|\Sigma|^{1-\lfloor d/2\rfloor})$ over the choice of $h_1$, the
double tabulation function $h_2\circ h_1$ is
fully random over $S$.
\end{theorem}
\cref{thm:double-uniform} should be contrasted by
the main theorem from \cite{thorup13doubletab}:

\begin{theorem}[Thorup {\cite{thorup13doubletab}}]\label{thm:double-thorup}
If $d\geq 6c$, then with probability
$1-o(|\Sigma|^{2-d/(2c)})$ over the choice of $h_1$, the
double tabulation function $h_2\circ h_1$ is
$k=|\Sigma|^{1/(5c)}$-independent.
\end{theorem}
\sd{The contrast here is, informally, that} \Cref{thm:double-uniform} is a statement about any one large set,
\Cref{thm:double-thorup} holds for all small sets. Also,
\cref{thm:double-uniform} with $d=4$ ``derived'' characters gets essentially
the same error probability as \cref{thm:double-thorup} with $d=6c$. Of
course, with $d=6c$, we are likely to get both properties with the same
double tabulation function.

Siegel \cite{siegel04hash} has proved that
with space $|\Sigma|$ it is impossible to get independence higher
than $|\Sigma|^{1-\Omega(1)}$ with constant time evaluation. This is much less than
the size of $S$ in \cref{thm:double-uniform}.

\Cref{thm:double-uniform} provides an extremely simple $O(n)$ space
implementation of a constant time hash function that is likely uniform
on any given set $S$ of size $n$. This should be compared with the
corresponding linear space uniform hashing of Pagh and Pagh \cite[\S 3]{PP08}.
Their
original implementation used Siegel's \cite{siegel04hash} highly
independent hash function as a subroutine. Dietzfelbinger and Woelfel
\cite{dietzfel03tabhash} found a simpler subroutine that was
not highly independent, but still worked in the uniform hashing from
 \cite{PP08}. However, Thorup's highly
independent double tabulation from \cref{thm:double-thorup}
is even simpler, providing us the simplest known implementation
of the uniform hashing in \cite{PP08}. However, as
discussed earlier, double tabulation uses many more derived characters
for high independence than for uniformity on a given set, so
for linear space uniform hashing on a given set, it is much
faster and simpler to use the double tabulation of \cref{thm:double-uniform}
directly
rather than \cite[\S 3]{PP08}.
We note that \cite[\S 4]{PP08} presents
a general trick to reduce the space from $O(n(\lg
n+\lg|\cR|))$ bits downto $(1+\eps)n\lg|\cR| + O(n)$ bits,
preserving the constant evaluation time. This reduction can also be
applied to Theorem \ref{thm:double-uniform} so that we also get
a simpler overall construction for a succinct dictionary
using $(1+\eps)n\lg|\cR| + O(n)$ bits of space and constant evaluation time.

We note that our analysis of \Cref{thm:uniform2}
does not apply to Pagh and Pagh's construction in \cite{PP08}, without strong
assumptions on the hash functions used, as we rely heavily on the
independence of output bits provided by simple tabulation.

\sd{
\paragraph{Peelable hash functions and invertible bloom filters}
Our proof of \Cref{thm:double-uniform} uses Thorup's variant
\cite{thorup13doubletab} of Siegel's notion of peelability \cite{siegel04hash}.
The hash function $h_1$ is a \sd{fully peelable} map of $S$ if for every subset
$Y\subseteq S$ there exists a key $y\in Y$ such that $h_1(y)$ has a unique
output character. If $h_1$ is peelable over $S$ and $h_2$ is a random simple
tabulation hash function, then $h_2\circ h_1$ is a uniform hash function over
$S$. \Cref{thm:double-uniform} thus follows by proving the following theorem.
\begin{theorem}\label{thm:peelable}
    Let $h : \Sigma^c\to \Sigma^d$ be a simple tabulation hash function and let
    $X$ be a set of keys with $|X|\le |\Sigma|/(1+\Omega(1))$. Then $h$ is
    fully peelable on $X$ with probability $1 - O(|\Sigma|^{1-\floor{d/2}})$.
\end{theorem}
The peelability of $h$ is not only relevant for uniform hashing. This property is
also critical for the hash function in Goodrich and Mitzenmacher's Invertible
Bloom Filters \cite{Goodrich11ibt}, which have found numerous applications in
streaming and data bases
\cite{eppstein2011s,eppstein2011straggler,mitzenmacher2012biff}.
So far Invertible Bloom Filters have been implemented
with fully random hashing, but \Cref{thm:peelable} states that simple
tabulation suffices for the underlying hash function.
}

\paragraph{Constant moments}
An alternative to Chernoff bounds in providing good concentration is to use
bounded moments. We show that the $k$th moment of simple tabulation comes
within a constant factor of that achieved by truly random hash functions for
any constant $k$.

\begin{theorem}\label{thm:kthmoment}
    Let $h : [u] \to \calR$ be a simple tabulation hash function.
    Let $x_0,\ldots,x_{m-1}$ be $m$ distinct keys from $[u]$ and let
    $Y_0,\ldots, Y_{m-1}$ be any random variables such that $Y_i\in [0,1]$ is a
    function of $h(x_i)$ with mean $\Ep{Y_i} = p$ for all $i\in [m]$.
    Define $Y = \sum_{i \in [m]} Y_i$ and $\mu = \Ep{Y} = mp$. Then for any
    constant integer $k \ge 1$:
    \[
        \Ep{( Y - \mu)^{2k}}
        =
        O\!\left (
            \sum_{j=1}^k \mu^j
        \right )\ ,
    \]
    where the constant in the $O$-notation is dependent on $k$ and $c$.
\end{theorem}

\subsection{Notation}\label{sec:notation}
Let $S\subseteq [u]$ be a set of keys. Denote by $\pi(S,i)$ the projection of $S$
on the $i$th character, i.e.~$\pi(S,i) = \{x_i | x\in S\}$. We also use this
notation for keys, so $\pi((x_0,\ldots,x_{c-1}),i) = x_i$. A \emph{position
character} is an element of $[c]\times \Sigma$. Under this definition a key
$x\in[u]$ can be viewed as a set of $c$ position characters
$\{(0,x_0),\ldots,(c-1,x_{c-1})\}$. Furthermore, for simple tabulation, we assume that $h$ is defined on
position characters as $h((i,\alpha)) = T_i[\alpha]$. This definition extends
to sets of position characters in a natural way by taking the XOR over the hash
of each position character. We denote the symmetric difference of the position
characters of a set of keys $x_1,\ldots, x_k$ by
\[
    \bigoplus_{i=1}^k x_k\ .
\]
We say that a set of keys $x_1,\ldots, x_k$ are \emph{independent} if their
corresponding hash values are independent. If the keys are not independent we
say that they are \emph{dependent}.

The \emph{hash graph} of hash functions $h_1 : [u]\to\calR_1,\ldots, h_k :
[u]\to\calR_k$ and a set $S\subseteq [u]$ is the graph in which each element of
$\calR_1\cup\ldots\cup\calR_k$ is a node, and the nodes are connected by the
(hyper-)edges $(h_1(x),\ldots,h_k(x)), x\in S$. In the graph there is a
one-to-one correspondence between keys and edges, so we will not distinguish
between those.

\subsection{Contents}
The paper is structured as follows. In \Cref{sec:minwise} we show how
\Cref{thm:uniform2} can be used to prove \Cref{thm:minwise} noting that the
same argument can be used to prove \Cref{thm:distinct}.
\Cref{sec:bound_dep,sec:uniform,app:uniform_multip} detail the proof of
\Cref{thm:uniform2}, which is the main technical part of the paper. Finally In
\Cref{sec:kthmoment} we prove \Cref{thm:kthmoment}.


%% file: minhash.tex
\section{MinHash with mixed tabulation}\label{sec:minwise}
In this section we prove \Cref{thm:minwise}. \Cref{thm:distinct} can be
proved using the same method. We will use the following lemma, which is
proved at the end of this section.

\begin{lemma}\label{thm:deviation}
    Let $h$ be a mixed tabulation hash function, $X \subset [u]$, and $Y$
    defined as in \cref{thm:uniform2} such that
    $\Ep{\abs{Y}} \in \left [ \frac{\abs{\Sigma}}{8}, \frac{\abs{\Sigma}}{4}
        \right )$.
    Then with probability $1 - \tilde{O}\left ( \abs{\Sigma}^{1-\floor{d/2}} \right )$
    \[
        \abs{Y} \in \Ep{\abs{Y}} \cdot \left ( 1 \pm O \left (
        \sqrt{\frac{\log \abs{\Sigma} \cdot \left ( \log \log \abs{\Sigma} \right )^2}{\abs{\Sigma}}} \right ) \right )
    \]
\end{lemma}

We are given sets $R$ and $B$ of $n_R$ and $n_B$ red and blue balls
respectively. Recall that the hash value $h(x)$ of a key $x$ is split into two
parts: one telling which of the $k$ bins $x$ lands in (i.e. the first
$\ceil{\lg k}$ bits) and the \emph{local hash value} in $[0,1)$ (the rest of
the bits).

\sd{Recall that} $|R|+|B| > |\Sigma|/2$ and assume that $|B|\ge
|R|$, as the other case is symmetric. For $C=R,B$, we define the set $S_C$ to
be the keys in $C$, for which the first $\ell_C$ bits of the local hash value
are $0$. We pick $\ell_C$ such that
\[
    \Ep{|S_C|} = 2^{-\ell_C}|C|\in\left(\frac{|\Sigma|}{8},
    \frac{|\Sigma|}{4}\right]\ .
\]
This is illustrated in \Cref{fig:minwise}.
We also define $X$ to be the keys of $R$ and $B$ whose first $\ell_B$ bits of
the local hash value are $0$.
\begin{figure}[htbp]
    \centering
    \includegraphics[width=.45\textwidth]{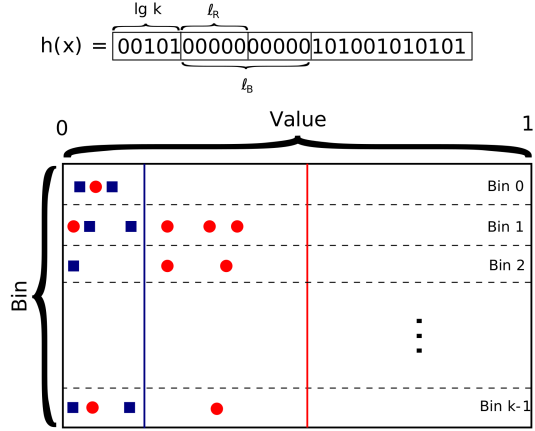}
    \caption{Illustration of the analysis for minwise hashing with mixed tabulation.
    	Since there are more red than blue balls, $\ell_R$ is smaller than
        $\ell_B$, illustrated by the blue vertical line being before the red one.
    	}
    \label{fig:minwise}
\end{figure}

We only bound the probability $P = \Prp{X^{\mathcal{M}} \ge (1+\delta)f}$ and
note that we can bound $\Prp{X^{\mathcal{M}}\le (1-\delta)f}$ similarly.
Consider also the alternative experiment $\ol X^{\cR}$ as defined in the
theorem. We let $\eps = c_0 \cdot \sqrt{\frac{\log \abs{\Sigma} \log \log \abs{\Sigma}}{\abs{\Sigma}}}$
for some large enough constant $c_0$. The set of $\floor{(1+\eps)\abs{R}}$ and
$\ceil{(1-\eps)\abs{B}}$ balls in this experiment is denoted $R'$ and $B'$
respectively. We define $S'_R$ and $S'_B$ to be the keys from $R'$ and $B'$
where the first $\ell_R$ and $\ell_B$ bits of the hash values are $0$
respectively.

In order to do bound $P$ we consider the following five \emph{bad events}:
\begin{itemize}
    \item [$E_1$:] $\abs{S_R} > \abs{S'_R}$.
	\item [$E_2$:] The remaining $\lg |\calR| - \ell_R$ output bits are fully
        independent when restricting $h$ to the keys of $S_R$.
    \item [$E_3$:] $\abs{S_B} < \abs{S'_B}$.
	\item [$E_4$:] The remaining $\lg |\calR| - \ell_B$ output bits are fully
        independent when restricting $h$ to the keys of $X$.
	\item [$E_5$:] There exists a bin which contains no key from $X$.
\end{itemize}
We will show that $\Prp{E_i} = \terrterm$ for $i=1,\ldots,5$. For $i=2,4$ this
is an immediate consequence of \Cref{thm:uniform2}. For $i=1,3$ we use
\Cref{thm:deviation} and let $c_0$ be sufficiently large. For $i=5$ we see
that if $E_3$ and $E_4$ do not occur then the probability that there exist
a bin with no balls from $X$ is at most:
\[
    k \cdot \left ( 1 - \frac{1}{k} \right )^{\abs{\Sigma}/8 \cdot (1-\eps)}
    \le
    k \cdot
    \exp \left ( -\frac{\abs{\Sigma}}{8k}(1-\eps) \right )
    \le
    k \cdot
    \exp \left ( -\frac{d\log \abs{\Sigma}}{2}(1-\eps) \right )
    \le
    O\!\left (
    \abs{\Sigma}^{1-d/2}
    \right )
\]
Hence by a union bound $\Prp{E_1 \cup \ldots \cup E_5} = \terrterm$ and:
\begin{align}
    \label{eq:plusErrTerm}
    P \le \Prp{X^{\mathcal{M}}\ge (1+\delta)f\cap \neg E_1\cap\ldots\cap\neg E_5} + \terrterm
\end{align}
Fix the $\ell_R$ bits of the hash values that decide $S_R,S'_R$ such that
$\abs{S_R} = a, \abs{S'_R} = a'$ and consider the probabilities
\begin{align*}
    P_1 & = \Prp{X^{\mathcal{M}}\ge (1+\delta)f\cap \neg E_1\cap\ldots\cap\neg E_5 \cond
    \left ( \abs{S_R} = a, \abs{S'_R} = a' \right )} \\
    P_2 & = \Prp{\ol X^{\cR} \ge (1+\delta)f \cond \left ( \abs{S_R} = a, \abs{S'_R} = a' \right )}
\end{align*}
We will now prove that $P_1 \le P_2$. This is trivial when $a > a'$ since $P_1 = 0$ in
this case so assume that $a \le a'$. We define $X'$ analogously to $X$ and
let $Y = X \cap S_R, Y' = X' \cap S'_R$.
Now fix the joint distribution of $(Y,Y')$ such that either
$E_2$ or $\abs{Y} \le \abs{Y'}$ with
probability $1$. We can do this without changing the marginal distributions of $Y,Y'$
since if $E_2$ doesn't occur the probability that $\abs{Y} \le i$ is at most the probability
that $\abs{Y'} \le i$ for any $i \ge 0$.
Now we fix the $\ell_B - \ell_R$ bits of the hash values that decide $X$ and $X'$. Unless
$E_2$ or $E_3$ happens we know that $\abs{Y} \le \abs{Y'}$ and $\abs{S_B} \ge \abs{S'_B}$.
Now assume that none of the bad events happen. Then we must have that
the probability that $X^{\mathcal{M}}\ge (1+\delta)f$ is no
larger than the probability that $\ol X^{\cR} \ge (1+\delta)f$. Since this is the case for
any choice of the $\ell_B - \ell_R$ bits of the hash values that decide $X$ and $X'$ we
conclude that $P_1 \le P_2$. Since this holds for any $a$ and $a'$:
\[
    \Prp{X^{\mathcal{M}}\ge (1+\delta)f\cap \neg E_1\cap\ldots\cap\neg E_5}
    \le
    \Prp{\ol X^{\cR} \ge (1+\delta)f}
\]
Inserting this into \eqref{eq:plusErrTerm} finishes the proof.

\input{antal-lemma}

%% file: antal-lemma.tex
\subsection{Proof of \Cref{thm:deviation}}\label{app:deviation}

We only prove the upper bound as the lower bound is symmetric.

Let $p_1,\ldots,p_b$ and $v_1,\ldots,v_b$ be the bit positions and bit values respectively
such that $Y$ is the set of keys $x \in X$ where $h(x)_{p_i} = v_i$ for all $i$.

Let $n = \abs{X}$, then $n2^{-b} \in I$, where
$I = \left [ \frac{\abs{\Sigma}}{8}, \frac{\abs{\Sigma}}{4} \right )$.
Partition $X$ into $2^b$ sets $X_0^0,\ldots,X_{2^b-1}^0$ such that $\abs{X_i^0} \in I$ for all $i \in [2^b]$.

For each $j = 1,\ldots,b$ and $i \in [2^{b-j}]$ let $X_i^j$ be the set of keys
$x \in \bigcup_{k=2^j \cdot i}^{2^j \cdot (i+1)-1} X_k^0$ where $h(x)_{p_k} = v_k$ for
$k = 1,\ldots,j$.
Equivalently, $X_i^j$ is the set of keys $x \in X_{2i}^{j-1} \cup X_{2i+1}^{j-1}$ where
$h(x)_{p_j} = v_j$. We note that $\Ep{\abs{X_i^j}} \in I$ and $X_0^b = Y$.

Let $A_j$ be the event that there exists $i \in [2^{b-j}]$ such that when the bit positions
$p_1,\ldots,p_{j-1}$ are fixed and the remaining bit positions of the keys in $X_i^j$ do not
hash independently.
By \cref{thm:uniform2} $\Prp{A_j} = O \left ( 2^{b-j}\abs{\Sigma}^{1-\floor{d/2}} \right )$.
Let $s_j = \sum_{i = 0}^{2^{b-j}-1} \abs{X_i^j}$.

Fix $j \in \set{1,2,\ldots,b}$ and the bit positions $p_1,\ldots,p_{j-1}$ of $h$
and assume that $A_{j-1}$ does not occur. Fix $i$ and say that $X_{i}^{j-1}$ contains
$r$ keys and write $X_{i}^{j-1} = \set{a_0,\ldots,a_{r-1}}$. Let $V_k$ be the random
variable defined by $V_k = 1$ if $h(a_k)_{p_j} = b_j$ and $V_k = 0$ otherwise. Let
$V = \sum_{k=0}^{r-1} V_k$. Then $V$ has mean $\frac{r}{2}$ and is the sum of
independent $0$-$1$ variables so by Chernoff's inequality:
\[
    \Prp{V \ge \frac{r}{2} \cdot (1 + \delta)} \le e^{-\delta^2 \cdot r/6}
\]
for every $\delta \in [0,1]$. Letting $\delta = \sqrt{\frac{6 d \log \abs{\Sigma}}{r}}$ we see
that with $\alpha = \sqrt{\frac{3}{2} d \log \abs{\Sigma}}$:
\[
    \Prp{V \ge \frac{r}{2} + \sqrt{r} \cdot \alpha} \le \abs{\Sigma}^{-d}
\]
We note that $V = \abs{X_i^{j-1} \cap X_{\floor{i/2}}^j}$. Hence we can rephrase it as:
\[
    \Prp{\abs{X_i^{j-1} \cap X_{\floor{i/2}}^j} \ge
         \frac{\abs{X_i^{j-1}}}{2} + \sqrt{\abs{X_i^{j-1}}} \cdot \alpha} \le \abs{\Sigma}^{-d}
\]
Now unfix $i$. By a union bound over all $i$ we see that with probability
$\ge 1 - 2^{b-j+1}\abs{\Sigma}^{-d}$ if $A_{j-1}$ does not occur:
\begin{align}
    \label{sjbound}
    s_j \le \sum_{i=0}^{2^{b-j+1}-1} \frac{\abs{X_i^{j-1}}}{2} + \sqrt{\abs{X_i^{j-1}}} \cdot \alpha
    \le
    \frac{s_{j-1}}{2} + \sqrt{2^{b-j+1}s_{j-1}} \cdot \alpha
\end{align}
Since $A_{j-1}$ occurs with probability $O \left ( 2^{b-j}\abs{\Sigma}^{1-\floor{d/2}} \right )$
we see that \eqref{sjbound} holds with probability $1 - O \left ( 2^{b-j}\abs{\Sigma}^{1-\floor{d/2}} \right )$.
Let $t_j = s_j2^{-b+j-1}$. Then \eqref{sjbound} can be rephrased as
\begin{align*}
    t_j \le t_{j-1} + \sqrt{t_{j-1}} \cdot \alpha
    \le \left ( \sqrt{t_{j-1}} + \frac{\alpha}{2} \right )^2
\end{align*}
Note that in particular:
\begin{align}
    \label{tjbound}
    \sqrt{t_j} \le \sqrt{t_{j-1}} + \frac{\alpha}{2}
\end{align}
Now assume that \eqref{sjbound} holds for every $j = b'+1,\ldots,b$ for some parameter $b'$
to be determined. This happens with probability
$1 - O \left ( 2^{b-b'}\abs{\Sigma}^{1-\floor{d/2}} \right )$. By \eqref{tjbound} we see
that $\sqrt{t_b} \le \sqrt{t_{b'}} + \frac{b-b'}{2}\alpha$. Hence:
\begin{align}
    \label{eq:finalsbbound}
    s_b \le \left ( \sqrt{s_{b'}2^{b'-b}} + \frac{b-b'}{\sqrt{2}}\alpha \right )^2
    = s_{b'}2^{b'-b} + \sqrt{2^{b'-b+1} s_{b'}}(b-b')\alpha +
    \left ( \frac{b-b'}{\sqrt{2}}\alpha \right )^2
\end{align}
We now consider two cases, when $n \le \Sigma \log^{2c} \Sigma$ and when $n > \Sigma \log^{2c} \Sigma$. First
assume that $n \le \Sigma \log^{2c} \Sigma$. Then we let $b' = 0$ and see that with probability
$1 - \terrterm$:
\[
    \abs{Y} = s_b \le
    \Ep{\abs{Y}} + \sqrt{2\Ep{\abs{Y}}}b \alpha + \left ( \frac{b}{\sqrt{2}}\alpha \right )^2
    =
    \Ep{\abs{Y}} +
    O\!\left ( \sqrt{\frac{\log \Sigma \left ( \log \log \Sigma \right )^2}{\Sigma}} \right )
\]
Where we used that $b = O\!\left(\log \log \Sigma \right)$. This proves the
claim when $n \le \Sigma \log^{2c} \Sigma$.

Now assume that $n > \Sigma \log^{2c} \Sigma$. In this case we will
use \Cref{thm:simple-concentration} below.
\begin{theorem}[P\v{a}tra\c{s}cu and Thorup \cite{patrascu12charhash}]
\label{thm:simple-concentration}
    If we hash $n$ keys into $m\leq n$ bins with simple tabulation, then,
    with high probability (whp.)\footnote{With probability $1 - n^{-\gamma}$
    for any $\gamma = O(1)$.}, every bin gets $n/m+O(\sqrt{n/m}\log^c n)$ keys.
\end{theorem}
Let $b' \ge 0$ be such that:
\[
    2^{-b'} = \Theta\!\left(\frac{\Sigma \cdot \log^{2c} n}{n} \right)
\]
With $\gamma = \floor{d/2}-1$ in \Cref{thm:simple-concentration} we see that with
probability $1 - \errterm$:
\begin{align}
    \label{eq:boundfromsimple}
    s_{b'} \le 2^{-b'}n + O\!\left(\sqrt{2^{-b'}n}\log^c n\right) =
    2^{-b'}n \cdot \left ( 1 + O\left(\sqrt{\frac{1}{\Sigma}}\right) \right )
\end{align}
By a union bound both \eqref{eq:finalsbbound} and \eqref{eq:boundfromsimple} hold with probability
$1 - \terrterm$ and combining these will give us the desired upper bound. This
concludes the proof when $n > \Sigma \log^{2c} \Sigma$.

%% file: bounddep_full.tex
\section{Bounding dependencies}\label{sec:bound_dep}
In order to proof our main technical result of \Cref{thm:uniform2} we need the
following structural lemmas regarding the dependencies of simple tabulation.

Simple tabulation is not 4-independent which means that there exists keys
$x_1,\ldots,x_4$, such that
$h(x_1)$ is dependent of $h(x_2), h(x_3), h(x_4)$. It was shown in
\cite{patrascu12charhash}, that for every $X\subseteq U$ with $|X| =
n$ there are at most $O(n^2)$ such dependent 4-tuples $(x_1,x_2,x_3,x_4)\in
X^4$.

In this section we show that a similar result holds in the case of dependent
$k$-tuples, which is one of the key ingredients in the proofs of the main
theorems of this paper.

We know from \cite{thorup12kwise} that if the keys $x_1,\ldots,x_k$ are dependent,
then there exists a non-empty subset $I\subset \{1,\ldots,k\}$ such that
\[
    \bigoplus_{i\in I} x_i = \emptyset\ .
\]
Following this observation we wish to bound the number of tuples which have
symmetric difference $\emptyset$.

\begin{lemma}\label{zeroSum}
    Let $X\subseteq U$ with $|X| = n$ be a subset. The number of $2t$-tuples
    $(x_1,\ldots, x_{2t})\in X^{2t}$ such that
    \[
        x_1\oplus \cdots \oplus x_{2t} = \emptyset
    \]
    is at most $((2t-1)!!)^c n^t$, where $(2t-1)!! = (2t-1)(2t-3)\cdots 3\cdot
    1$.
\end{lemma}

It turns out that it is more convenient to prove the following more
general lemma.
\begin{lemma}
\label{zeroSumProvable}
Let $A_1,\ldots,A_{2t} \subset U$ be sets of keys. The number of $2t$-tuples $(x_1, \ldots, x_{2t}) \in A_1 \times \cdots \times A_{2t}$ such that
\begin{equation}
\label{eq:zeroSumProvable}
x_1 \oplus \cdots \oplus x_{2t} = \emptyset
\end{equation}
is at most $((2t-1)!!)^c \prod_{i=1}^{2t} \sqrt{\abs{A_i}}$.
\end{lemma}
\begin{proof}[Proof of \cref{zeroSumProvable}]
Let $(x_1, \ldots, x_{2t})$ be such a $2t$-tuple.
\Cref{eq:zeroSumProvable} implies that the number of times each position character
appears is an even number. Hence we can partition $(x_1,\ldots,x_{2t})$ into $t$ pairs
$(x_{i_1},x_{j_1}), \ldots, (x_{i_t},x_{j_t})$ such that
$\pi(x_{i_k},c-1) = \pi(x_{j_k},c-1)$ for $k=1,\ldots,t$.
Note that there are at $(2t-1)!!$ ways to partition the elements in such a way.
This is illustrated in \Cref{fig:dependency}.
\begin{figure}[htbp]
    \centering
    \includegraphics[width=.3\textwidth]{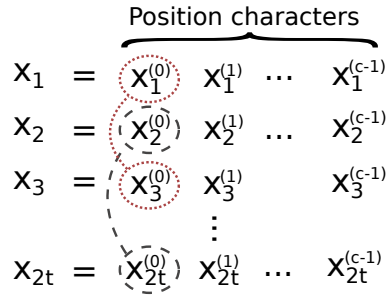}
    \caption{Pairing of the position characters of $2t$ keys. $x_1^{(0)}$
    can be matched to $2t-1$ position characters, $x_2^{(0)}$ to $2t-3$,
    etc.}
    \label{fig:dependency}
\end{figure}

We now prove the claim by induction on $c$. First assume that $c=1$.
We fix some partition $(x_{i_1},x_{j_1}), \ldots, (x_{i_t},x_{j_t})$ and count the
number of $2t$-tuples which fulfil  $\pi(x_{i_k},c-1) = \pi(x_{j_k},c-1)$ for $k=1,\ldots,t$.
Since $c=1$ we have $x_{i_k},x_{j_k} \in A_{i_k} \cap A_{j_k}$. The number of ways to choose such
a $2t$-tuple is thus bounded by:
\[
\prod_{k=1}^t \abs{A_{i_k} \cap A_{j_k}} \le
\prod_{k=1}^t \min\set{\abs{A_{i_k}}, \abs{A_{j_k}}} \le
\prod_{k=1}^t \sqrt{\abs{A_{i_k}} \abs{A_{j_k}}} =
\prod_{k=1}^{2t} \sqrt{\abs{A_k}}
\]
And since there are $(2t-1)!!$ such partitions the case $c=1$ is finished.

Now assume that the lemma holds when the keys have $< c$ characters.
As before, we fix some partition $(x_{i_1},x_{j_1}), \ldots, (x_{i_t},x_{j_t})$ and
count the number of $2t$-tuples which satisfy $\pi(x_{i_k},c-1) = \pi(x_{j_k},c-1)$
for all $k=1,\ldots,t$.
Fix the last position character $(a_k,c-1) = \pi(x_{i_k},c-1) = \pi(x_{j_k},c-1)$
for $k=1,\ldots,t$, $a_k \in \Sigma$.
The rest of the position characters from $x_{i_k}$ is then from the
set
\[
A_{i_k}[a_k] =
\set{
x \backslash (a_k,c-1)
\mid
(a_k,c-1) \in x, x\in A_{i_k}
}
\]
By the induction hypothesis the number of ways to choose $x_1,\ldots,x_{2t}$ with
this choice of $a_1,\ldots,a_t$ is then at most:
\[
((2t-1)!!)^{c-1}
\prod_{k=1}^{t}
\sqrt{\abs{A_{i_k}[a_k]} \abs{A_{j_k}[a_k]} }
\]
Summing over all choices of $a_1,\ldots,a_t$ this is bounded by:
\begin{align*}
& ((2t-1)!!)^{c-1}
\sum_{a_1,\ldots,a_t \in \Sigma}
\prod_{k=1}^{t}
\sqrt{\abs{A_{i_k}[a_k]} \abs{A_{j_k}[a_k]} }
\\
= &
((2t-1)!!)^{c-1}
\prod_{k=1}^{t}
\sum_{a_k \in \Sigma}
\sqrt{\abs{A_{i_k}[a_k]} \abs{A_{j_k}[a_k]} }
\\
\le &
((2t-1)!!)^{c-1}
\prod_{k=1}^{t}
\sqrt{\sum_{a_k \in \Sigma} \abs{A_{i_k}[a_k]}}
\sqrt{\sum_{a_k \in \Sigma} \abs{A_{j_k}[a_k]}}
\numberthis\label{eq:cauchy}
\\
= &
((2t-1)!!)^{c-1}
\prod_{k=1}^{t}
\sqrt{\abs{A_{i_k}}}
\sqrt{\abs{A_{j_k}}}
=
((2t-1)!!)^{c-1}
\prod_{k=1}^{2t}
\sqrt{\abs{A_k}}
\end{align*}
Here \eqref{eq:cauchy} is an application of Cauchy-Schwartz's inequality.
Since there are $(2t-1)!!$ such partitions the conclusion follows.
\end{proof}

%
%
%

%% file: uniform.tex
\section{Uniform hashing in constant time}\label{sec:uniform}

This section is dedicated to proving \cref{thm:uniform2}. We will show the
following more general theorem. \sd{This proof also implies the result of
\Cref{thm:peelable}.}

\begin{theorem}\label{thm:uniform3}
    Let $h=h_2\circ h_1^\star$ be a mixed tabulation hash function.
    Let $X\subset [u]$ be any input set. For each $x\in X$, associate a function
    $f_x : \calR\to \{0,1\}$. Let $Y = \{x\in X\cond f_x(h(x)) = 1\}$ and assume $\Ep{|Y|}\le
    |\Sigma|/(1+\eps)$.

    Then the keys of $h_1^\star(Y)\subseteq \Sigma^{c+d}$
    are \sd{peelable} with probability $1 - O(|\Sigma|^{1-\floor{d/2}})$.
\end{theorem}
Here, we consider only the case when there exists a $p$ such that
$\Prp{f_x(z) = 1} = p$ for all $x$, when $z$ is uniformly distributed in $\calR$.
In \Cref{app:uniform_multip} we \sd{sketch} the details when this is
not the case. We note that the full proof uses the same ideas but is more
technical.

The proof is structured in the following way: (1) We fix $Y$ and assume the
key set $h_1^\star(Y)$ is not
independent. (2) With $Y$ fixed this way we construct a \emph{bad event}. (3)
We unfix $Y$ and show that the probability of a bad event occurring is low using
a union bound. Each bad event consists of independent ``sub-events'' relating
to subgraphs of the hash graph of $h_1(Y)$. These sub-events fall into four
categories, and for each of those we will bound the probability that the event
occurs.

First observe that if a set of keys $S$ consists of independent keys, then the
set of keys $h_1^\star(S)$ are also independent.

We will now describe what we mean by a \emph{bad event}. We consider the hash
function $h_1 : [u]\to\Sigma^d$ as $d$ simple tabulation hash functions
$h^{(0)}, \ldots, h^{(d-1)} : [u]\to \Sigma$ and define $G_{i,j}$ to be the
hash graph of $h^{(i)}, h^{(j)}$ and the input set $X$.

Fix $Y$ and consider some $y\in Y$. If for some $i,j$, the component of
$G_{i,j}$ containing $y$ is a tree, then we can perform a peeling process and
observe that $h_1^\star(y)$ must be independent of $h_1^\star(Y\sm \{y\})$.
Now assume that there exists some $y_0\in Y$ such that $h_1^\star(y_0)$ is dependent of
$h_1^\star(Y\sm\{y_0\})$, then $y_0$ must lie on a (possibly empty) path leading to a
cycle in each of $G_{2i,2i+1}$ for $i\in [\floor{d/2}]$. We will call such a path
and cycle a \emph{lollipop}. Denote this lollipop
by $y_0,y_1^i,y_2^i,\ldots,y_{p_i}^i$. For each such $i$ we will construct a
list $L_i$ to be part of our bad event. Set $s\stackrel{def}{=}
\ceil{2\log_{1+\eps}\abs{\Sigma}}$.
The list $L_i$ is constructed in the following manner: We
walk along $y_1^i,\ldots,y_{p_i}^i$ until we meet an \emph{obstruction}.
Consider a key $y_j^i$. We will say that $y_j^i$ is an obstruction if it falls
into one of the following four cases as illustrated in \cref{fig:save_cases}.
\begin{description}
    \item [A] There exists some subset
        $B\subseteq\{y_0,y_1^i,\ldots,y_{j-1}^i\}$ such that $y_j^i =
        \bigxor_{y\in B} y$.
    \item [B] If case \textbf{A} does not hold and there exists some subset
        $B\subseteq\{y_0,y_1^i,\ldots,y_{j-1}^i\}\cup L_0\cup\ldots\cup
        L_{i-1}$ such that $y_j^i = \bigxor_{y\in B} y$.
    \item [C] $j = p_i < s$ (i.e.~$y_j^i$ is the last key on the cycle). In
        this case $y_j^i$ must share a node with either $y_0$ (the path of the
        lollipop is empty) or with two of the other keys in the lollipop.
    \item [D] $j=s$. In this case the keys $y_1^i,\ldots,y_s^i$ form a path
        keys independent from $L_0,\ldots,L_{i-1}$.
\end{description}
In all four cases we set $L_i = (y_1^i,\ldots, y_j^i)$ and
we associate an attribute $A_i$. In case \textbf{A} we set $A_i
= B$. In case \textbf{B} we set $A = (x^{(0)},\ldots,x^{(c-1)})$,
where $x^{(r)}\in B$ is chosen such that $\pi(y_j^i,r) = \pi(x^{(r)},r)$.
In \textbf{C} we set $A_i = z$, where $z$ is the smallest value such that
$y_z^i$ shares a node with $y_j^i$, and in \textbf{D} we set $A_i = \emptyset$.
Denote the lists by $L$ and the types and attributes of the lists by $T, A$. We
have shown, that if there is a dependency among the keys of $h_1^\star(Y)$, then we
can find such a bad event $(y_0, L, T, A)$.
\begin{figure*}[htbp]
    \centering
    \includegraphics[width=.9\textwidth]{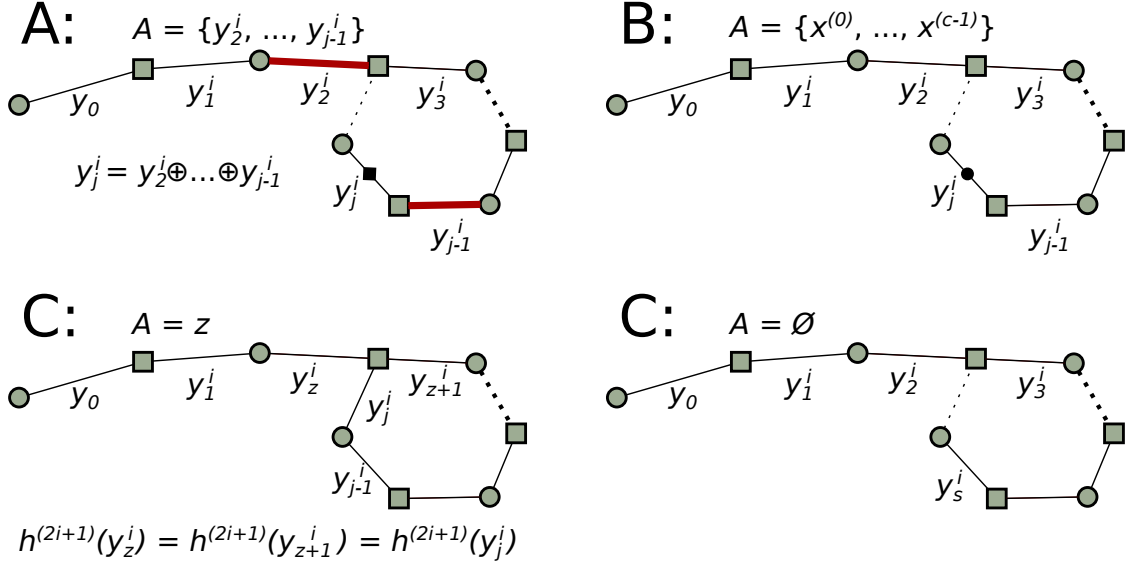}
    \caption{The four types of violations. Dependent keys are denoted by
    $\blacksquare$ and $\bullet$.}
    \label{fig:save_cases}
\end{figure*}

Now fix $y_0\in X, l = (l_0,\ldots,l_{\floor{d/2}-1})$.
Let $F(y_0,l)$ be the event that there exists a quadruple
$(y_0,L,T,A)$ forming a bad event such that $|L_i| = l_i$. We use the shorthand
$F = F(y_0,l)$. Let $F(y_0,L,T,A)$ denote the event that a given quadruple
$(y_0,L,T,A)$ occurs. Note that a quadruple $(y_0, L,T,A)$ only occurs if some
conditions are satisfied for $h_1$ (i.e.~that the hash graph forms the
lollipops as described earlier) and $h_2$ (i.e.~that the keys of the lollipops
are contained in $Y$). Let $F_1(y_0,L,T,A)$ and $F_2(y_0,L,T,A)$ denote the
the event that those conditions are satisfied, respectively. Then
\begin{align*}
    \Prp{F}
    &\le \sum_{\text{bad event $L,T,A$}} \Prp{F(y_0,L,T,A)} \\
    &= \sum_{\text{bad event $L,T,A$}} \Prp{F_2(y_0,L,T,A) |
        F_1(y_0,L,T,A)}\cdot \Prp{F_1(y_0,L,T,A)}\ .
\end{align*}
We note, that $F_1(y_0,L,T,A)$ consists of independent events for each
$G_{2i,2i+1}$, for $i\in[\floor{d/2}]$. Denote these restricted events by
$F_1^i(y_0,L,T,A)$.

For a fixed $h_1$ we can bound $\Prp{F_2(y_0,L,T,A)}$ in the following way: For
each $i\in[\floor{d/2}]$ we choose a subset $V_i\subseteq L_i$ such that
$S = \{y_0\}\bigcup_i V_i$ consists of independent keys. Since these keys are
independent, so is $h_1^\star(S)$, so we can bound the probability that
$S\subseteq Y$ by $p^{|S|}$. We can split this into one part for each $i$.
Define
\[
    p_i \stackrel{def}{=} p^{|V_i|}\cdot\Prp{F_1^i(y_0,L_i,T_i,A_i)}\ .
\]
We can then bound $\Prp{F}\le p\cdot\prod_{i\in[\floor{d/2}]} p_i$.

We now wish to bound the probability $p_i$.
Consider some $i\in [\floor{d/2}]$. We split the analysis into a case for each
of the four types:
\begin{description}
    \item [A] Let $\Delta(y_0)$ be the number of triples $(a,b,c)\in X^3$ such
        that $y_0\xor a\xor b\xor c = \emptyset$. Note that the size of the
        attribute $|A_i|\ge 3$ must be odd. Consider the following three cases:
        \begin{enumerate}
            \item $|A_i| = 3, y_0\in A_i$: We have $y_0$ is the $\xor$-sum
                of three elements of $L_i$. The number of ways this can happen
                (i.e.~the number of ways to choose $L_i$ and $A_i$)
                is bounded by $l_i^3n^{l_i-3}\Delta(y_0)$ -- The indices of the
                three summands can be chosen in at most $l_i^3$ ways, and the
                corresponding keys in at most $\Delta(y_0)$ ways.
                The remaining elements can be chosen in at most $n^{l_i-3}$ ways.
            \item $|A_i| \ge 5, y_0\in A_i$: By \cref{zeroSumProvable} we can
                choose $L_i$ and $A_i$ in at most $l_i^{O(1)}\cdot n^{l_i-5/2}$
                ways.
            \item $|A_i|\ge 3, y_0\notin A_i$: By \cref{zeroSumProvable} we can
                choose $L_i$ and $A_i$ in at most $l_i^{O(1)}\cdot n^{l_i - 2}$
                ways.
        \end{enumerate}
        To conclude, we can choose $L_i$ and $A_i$ in at most
        \[
            l_i^{O(1)}\cdot n^{l_i-2}\cdot\left(1 +
            \frac{\Delta(y_0)}{n}\right)
        \]
        ways. We can choose $V_i$ to be $L_i$ except for the last key. We note
        that $V_i\cup \{y_0\}$ form a path in $G_{2i,2i+1}$, which
        happens with probability $1/\abs{\Sigma}^{l_i-1}$ since the keys
        are independent. For type \textbf{A} we thus get the bound
        \begin{align*}
            p_i&\le l_i^{O(1)}\cdot p^{l_i-1}\cdot
            n^{l_i-2}\cdot\left(1+\frac{\Delta(y_0)}{n}\right)\cdot\frac{1}{\abs{\Sigma}^{l_i-1}}
            \le
            l_i^{O(1)}\cdot\left(1+\frac{\Delta(y_0)}{n}\right)\cdot\frac{1}{\abs{\Sigma}}\cdot\frac{p}{(1+\eps)^{l_i-2}}
            \\ &\le
            l^{O(1)}\cdot\left(1+\frac{\Delta(y_0)}{n}\right)\cdot\frac{1}{\abs{\Sigma}}\cdot\frac{1}{(1+\eps)^{l_i/2}}\
            .
        \end{align*}
    \item [B] All but the last key of $L_i$ are independent and can be chosen
        in at most $n^{l_i-1}$ ways. The last key is uniquely defined by $A_i$,
        which can be chosen in at most $l^c$ ways (where $l = \sum_i l_i)$,
        thus $L_i$ and $A_i$ can be chosen in at most $n^{l_i-1}l^c$ ways.
        Define $V_i$ to be all but the last key of $L_i$. The
        keys of $L_i\cup\{y_0\}$ form a path, and since the last key of
        $L_i$ contains a position character not in $V_i$, the probability of
        this path occurring is exactly $1/\abs{\Sigma}^{l_i}$, thus we get
        \[
            p_i \le l^c\cdot n^{l_i-1}\cdot p^{l_i-1}\cdot
            \frac{1}{\abs{\Sigma}^{l_i}}
            \le l^{O(1)}\cdot\frac{1}{\abs{\Sigma}}\cdot\frac{1}{(1+\eps)^{l_i-1}}
            \le l^{O(1)}\cdot\frac{1}{\abs{\Sigma}}\cdot\frac{1}{(1+\eps)^{l_i/2}}\ .
        \]
    \item [C] The attribute $A_i$ is just a number in $[l_i]$, and $L_i$ can be
        chosen in at most $n^{l_i}$ ways. We can choose $V_i = L_i$.
        $V_i\cup\{y_0\}$ is a set of independent keys forming a path leading to a
        cycle, which happens with probability $1/\abs{\Sigma}^{l_i+1}$, so we get the
        bound
        \[
            p_i\le l_i\cdot n^{l_i}\cdot p^{l_i}\cdot \frac{1}{\abs{\Sigma}^{l_i+1}}
            \le l_i\cdot \frac{1}{\abs{\Sigma}}\cdot\frac{1}{(1+\eps)^{l_i}}
            \le l^{O(1)}\cdot \frac{1}{\abs{\Sigma}}\cdot\frac{1}{(1+\eps)^{l_i/2}}\ .
        \]
    \item [D] The attribute $A_i = \emptyset$ is uniquely chosen. $L_i$
        consists of $s$ independent keys and can be chosen in at most $n^s$
        ways. We set $V_i = L_i$. We
        get
        \[
            p_i\le n^s\cdot p^s\cdot \frac{1}{\abs{\Sigma}^s} \le
            \frac{1}{(1+\eps)^s} \le
            \frac{1}{\abs{\Sigma}}\cdot\frac{1}{(1+\eps)^{l_i/2}}\ .
        \]
\end{description}

We first note, that there exists $y_0$ such that $\Delta(y_0) = O(n)$. We have
just shown that for a specific $y_0$ and partition of the lengths
$(l_0,\ldots,l_{\floor{d/2}})$ we get
\[
    \Prp{F}\le
    p\cdot\left(l^{O(1)}\cdot\frac{1}{\abs{\Sigma}}\right)^{\floor{d/2}}\cdot\frac{1}{(1+\eps)^{l/2}}\
        .
\]
Summing over all partitions of the $l_i$'s and choices of $l$ gives
\[
    \sum_{l\ge 1} p\cdot
    l^{O(1)}\cdot\abs{\Sigma}^{-\floor{d/2}}\cdot\frac{1}{(1+\eps)^{l/2}} \le
    O\left(p\cdot\abs{\Sigma}^{-\floor{d/2}}\right)\ .
\]
We have now bounded the probability for $y_0\in X$ that $y_0\in Y$ and $y_0$ is
dependent on $Y\sm\{y_0\}$. We relied on $\Delta(y_0) = O(n)$, so we cannot
simply take a union bound. Instead we note that, if $y_0$ is independent of
$Y\sm\{y_0\}$ we can peel $y_0$ away and use the same argument on
$X\sm\{y_0\}$. This gives a total upper bound of
\[
    O\!\left(\sum_{y_0\in X} p\cdot\abs{\Sigma}^{-\floor{d/2}}\right) =
    O(\abs{\Sigma}^{1-\floor{d/2}})\ .
\]
This finishes the proof.\qed

%% file: uniform_multip.tex
\section{Uniform hashing with multiple probabilities}
\label{app:uniform_multip}

Here we present a sketch in extending the proof in \cref{sec:uniform}.
We only need to change the proof where we bound $p_i$.
Define $p_x = \Prp{f_x(z) = 1}$ when $z$ is uniformly distributed in $\calR$.
First we argue that cases \textbf{B},
\textbf{C} and \textbf{D} are handled in almost the exact same way. In
the original proof we argued that for some size $v$ we can choose $V_i, \abs{V_i}=v$
in at most $n^{v}$
ways and for each choice of $V_i$ the probability that it is contained in $Y$
is at most $p^{v}$, thus multiplying the upper bound by
\[
    n^{v} p^{v} = \left ( \E \abs{Y} \right )^{v}
\]
For our proof we sum over all choices of $V_i$ and add the probabilities that
$V_i$ is contained in $Y$ getting the exact same estimate:
\[
    \sum_{V_i \in U, \abs{V_i} = v} \left ( \prod_{x \in V_i} p_x \right )
    \le
    \left ( \sum_{x \in U} p_x \right )^v
    =
    \left ( \E \abs{Y} \right )^v
\]
The difficult part is to prove the claim in case \textbf{A}.

For all $i \ge 0$ we set
\[
    n_i = \abs{\set{x\in X\cond p_x\in \left ( 2^{-i-1}, 2^{-i} \right ]}}\ .
\]
Now observe, that
$\sum_{i \ge 0} n_i 2^{-i} \le 2\abs{\Sigma}/(1+\eps) = O(\abs{\Sigma})$. Define $m_i =
\sum_{j \le i} n_j$, we then have:
\[
    \sum_{i \ge 0} m_i 2^{-i} =
    \sum_{i \ge 0} n_i \left ( \sum_{j \ge i} 2^{-j} \right ) =
    \sum_{i \ge 0} n_i 2^{-i+1} =
    O(\abs{\Sigma})
\]
We let $X_i = \set{x \in X \mid p_x > 2^{-i-1}}$ and note that $m_i = \abs{X_i}$.
For each $y_0 \in X$ we will define $\Delta'(y_0)$ (analogously to $\Delta(y_0)$) in the following way:
\[
    \Delta'(y_0) = \sum_{a,b,c \in X} \min\set{p_ap_b,p_bp_c,p_cp_a}
\]
where we only sum over triples $(a,b,c)$ such that $y_0 \xor a \xor b \xor c = \emptyset$.
Analogously to the original proof we will show that there exists $y_0$ such
that $\Delta'(y_0) \le O(\abs{\Sigma})$. The key here is to prove that:
\[
    \sum_{y_0 \in X} \Delta'(y_0) = O\left(n \abs{\Sigma}\right )
\]
Now consider a $4$-tuple $(y_0,a,b,c)$ such that $y_0 \xor a \xor b \xor c = \emptyset$.
Let $i \ge 0$ be the smallest non-negative integer such that $b,c \in X_i$. Then:
\[
    \min\set{p_ap_b,p_bp_c,p_cp_a} \le \min\set{p_b,p_c} \le 2^{-i}
\]
By \ref{zeroSumProvable} we see that for any $i$ there are at most
$O(n m_i)$ $4$-tuples $(y_0,a,b,c)$
such that $b,c \in X_i$. This gives the following bound on the total sum:
\[
    \sum_{y_0 \in X} \Delta'(y_0) \le
    \sum_{i \ge 0} O(n m_i) \cdot 2^{-i} =
    O\left(n \abs{\Sigma} \right )
\]
Hence there exists $y_0$ such that $\Delta'(y_0) = O(\abs{\Sigma})$ and we can finish
case \textbf{A.1} analogously to the original proof.

Now we turn to case \textbf{A.2} where $\abs{A_i} \ge 5, y_0 \in A_0$. We will here
only consider the case $\abs{A_i} = 5$, since the other cases follow by the same
reasoning. We will choose $V_i$ to consist of all of $L_i \sm A_i$ and $3$ keys
from $A_i$. We will write $A_i = \set{a,b,c,d,e}$ and find the smallest $\alpha,\beta,\gamma$
such that $a,b \in X_\alpha, c,d \in X_\beta, e \in X_\gamma$. Then:
\[
    \prod_{x \in V_i} p_x
    \le
    \left ( \prod_{x \in V_i \sm A_i} p_x \right )
    2^{-\alpha} 2^{-\beta} 2^{-\gamma}
\]
When $a,b \in X_\alpha, c,d \in X_\beta, e \in X_\gamma$ we can choose $a,b,c,d,e$
in at most $m_\alpha m_\beta \sqrt{m_\gamma}$ ways by \cref{zeroSumProvable}.
Hence, when we sum over
all choices of $V_i$ we get an upper bound of:
\[
    \left ( \sum_{x \in X} p_x \right )^{l_i-5}
    \left ( \sum_{\alpha,\beta,\gamma \ge 0} m_\alpha m_\beta \sqrt{m_\gamma}
            2^{-\alpha} 2^{-\beta} 2^{-\gamma} \right )
    =
    \left ( \sum_{x \in X} p_x \right )^{l_i-5}
    \left ( \sum_{\alpha \ge 0} m_\alpha 2^{-\alpha} \right )^2
    \left ( \sum_{\alpha \ge 0} \sqrt{m_\alpha} 2^{-\alpha} \right )
\]
Now we note that by Cauchy-Schwartz inequality:
\[
    \sum_{\alpha \ge 0} \sqrt{m_\alpha} 2^{-\alpha} \le
    \sqrt{\sum_{\alpha \ge 0} 2^{-\alpha}}
    \sqrt{\sum_{\alpha \ge 0} m_\alpha 2^{-\alpha}}
    = O(\sqrt{\abs{\Sigma}})
\]
Hence we get a total upper bound of $O(\abs{\Sigma}^{l_i-5/2})$ and we can finish
the proof in analogously to the original proof.

Case \textbf{A.3} is handled similarly to \textbf{A.2}.

%% file: kthmoment.tex
\section{Constant moment bounds}\label{sec:kthmoment}
This section is dedicated to proving \cref{thm:kthmoment}.

Consider first \cref{thm:kthmoment} and let $k = O(1)$ be fixed. Define $Z_i =
Y_i - p$ for all $i\in [m]$ and $Z = \sum_{i\in [m]} Z_i$. We wish to bound
$\Ep{Z^{2k}}$ and by linearity of expectation this equals:
\[
    \Ep{Z^{2k}}
    =
    \sum_{r_0,\ldots,r_{2k-1} \in [m]^{2k}}
    \Ep{Z_{r_0} \cdots Z_{r_{2k-1}}}
\]
Fix some $2k$-tuple $r = (r_0,\ldots,r_{2k-1}) \in [m]^{2k}$ and define
$V(r) = \Ep{Z_{r_0} \cdots Z_{r_{2k-1}}}$. Observe, that if there exists
$i \in [2k]$ such that $x_{r_i}$ is independent of
$(x_{r_j})_{j \neq i}$ then
\[
    V(r) =
    \Ep{Z_{r_0} \cdots Z_{r_{2k-1}}} =
    \Ep{Z_{r_i}}
    \Ep{\prod_{j \neq i} Z_{r_j}}
    = 0
\]

The following lemma bounds the number of $2k$-tuples, $r$, for which $V(r)\ne
0$.
\begin{lemma}\label{lem:moment-rtuples}
    The number of $2k$-tuples $r$ such that $V(r) \neq 0$ is $O(m^k)$.
\end{lemma}
\begin{proof}
Fix $r \in [m]^{2k}$ and let $T_0,\ldots,T_{s-1}$ be all subsets of $[2k]$
such that $\bigxor_{i \in T_j} x_{r_i} = \emptyset$ for $j \in [s]$. If
$\bigcup_{j \in [s]} T_j \neq [2k]$ we must have $V(r) = 0$ as there exists some
$x_{r_i}$, which is independent of $(x_{r_j})_{j \neq i}$. Thus we can assume
that $\bigcup_{j\in [s]} T_j = [2k]$.

Now fix $T_0,\ldots,T_{s-1}\subseteq [2k]$ such that $\bigcup_{j \in [s]} T_j
= [2k]$ and count the number of ways to choose $r \in [m]^{2k}$ such that
$\bigxor_{i \in T_j} x_{r_i} = \emptyset$ for all $j\in[s]$. Note that
$T_0,\ldots,T_{s-1}$ can be chosen in at most $2^{2k} = O(1)$ ways, so if
we can bound the number of ways to choose $r$ by $O(m^k)$ we are done.
Let $A_i = \bigcup_{j < i} T_j$ and $B_i = T_i \sm A_i$ for $i \in [s]$.
We will choose $r$ by choosing $(x_{r_i})_{i \in B_0}$, then
$(x_{r_i})_{i \in B_1}$, and so on up to $(x_{r_i})_{i \in B_{s-1}}$.
When we choose $(x_{r_i})_{i \in B_{j}}$ we have already chosen
$(x_{r_i})_{i \in A_{j}}$ and by \cref{zeroSumProvable} the number of ways
to choose $(x_{r_i})_{i \in B_{j}}$ is bounded by:
\[
    \left ( (\abs{T_j}-1)!! \right )^c
    m^{\abs{B_j}/2}
    = O \left ( m^{\abs{B_j}/2} \right )
\]

Since $\bigcup_{j\in [s]} B_j = [2k]$ we conclude that the number of ways to
choose $r$ such that $V(r) \ne 0$ is at most $O(m^k)$.
\end{proof}

We note that since $\abs{V(r)} \le 1$ this already proves that
\[
    \Ep{Z^{2k}} \le O(m^k)
\]

Consider now any $r \in [m]^{2k}$ and let $f(r)$ denote the size of the
largest subset $I\subset [2k]$ of independent keys $(x_{r_i})_{i\in I}$.
We then have
\[
    \abs{\Ep{\prod_{i \in [2k]} Z_{r_i}}} \le
    \Ep{\abs{\prod_{i \in [2k]} Z_{r_i}}} \le
    \Ep{\abs{\prod_{i \in I} Z_{r_i}}} \le
    O \left ( p^{f(r)} \right )
\]
We now fix some value $s\in\oneToN{2k}$ and count the number of $2k$-tuples
$r$ such that $f(r) = s$. We can bound this number by first choosing
the $s$ independent keys of $I$ in at most $m^s$ ways. For each remaining key
we can write it as a sum of a subset of $(x_{r_i})_{i\in I}$. There are at most
$2^s = O(1)$ such subsets, so there are at most $O(m^s)$ such $2k$-tuples $r$
with $f(r) = s$.

Now consider the $O(m^k)$ $2k$-tuples $r \in [m]^{2k}$ such that $V(r) \neq 0$.
For each $s \in \oneToN{2k}$ there is $O(m^{\min\set{k,s}})$ ways to choose
$r$ such that $f(r) = s$. All these choices of $r$ satisfy $V(r) \le O(p^s)$.
Hence:
\[
    \Ep{Z^{2k}} =
    \sum_{r \in [m]^{2k}} V(r)
    \le
    \sum_{s=1}^{2k} O(m^{\min\set{k,s}}) \cdot O(p^s)
    =
    O \left (
        \sum_{s=1}^{k} (pm)^s
    \right )\ .
\]
This finishes the proof of \cref{thm:kthmoment}.\qed

A similar argument can be used to show the following theorem, where the bin
depends on a query key $q$.
\begin{theorem}\label{thm:kthmomentq}
    Let $h : [u] \to \calR$ be a simple tabulation hash function.
    Let $x_0,\ldots,x_{m-1}$ be $m$ distinct keys from $[u]$ and let $q\in [u]$
    be a \emph{query} key distinct from $x_0,\ldots,x_{m-1}$.
    Let $Y_0,\ldots,Y_{m-1}$ be any random variables such that $Y_i\in [0,1]$ is
    a function of $(h(x_i),h(q))$ and for all $r \in \calR$,
    $\Ep{Y_i \mid h(q) = r} = p$ for all $i \in [m]$.
    Define $Y = \sum_{i \in [m]} Y_i$ and $\mu = \Ep{Y} = mp$. Then for any
    constant integer $k \ge 1$:
    \[
        \Ep{( Y - \mu)^{2k}}
        \le
        O\!\left (
            \sum_{j=1}^k \mu^j
        \right )\ ,
    \]
    where the constant in the $O$-notation is dependent on $k$ and $c$.
\end{theorem}

%% file: kpartition.bbl
\begin{thebibliography}{10}
\providecommand{\url}[1]{#1}
\csname url@samestyle\endcsname
\providecommand{\newblock}{\relax}
\providecommand{\bibinfo}[2]{#2}
\providecommand{\BIBentrySTDinterwordspacing}{\spaceskip=0pt\relax}
\providecommand{\BIBentryALTinterwordstretchfactor}{4}
\providecommand{\BIBentryALTinterwordspacing}{\spaceskip=\fontdimen2\font plus
\BIBentryALTinterwordstretchfactor\fontdimen3\font minus
  \fontdimen4\font\relax}
\providecommand{\BIBforeignlanguage}[2]{{%
\expandafter\ifx\csname l@#1\endcsname\relax
\typeout{** WARNING: IEEEtran.bst: No hyphenation pattern has been}%
\typeout{** loaded for the language `#1'. Using the pattern for}%
\typeout{** the default language instead.}%
\else
\language=\csname l@#1\endcsname
\fi
#2}}
\providecommand{\BIBdecl}{\relax}
\BIBdecl

\bibitem{thorup12kwise}
M.~Thorup and Y.~Zhang, ``Tabulation-based 5-independent hashing with
  applications to linear probing and second moment estimation,'' \emph{SIAM
  Journal on Computing}, vol.~41, no.~2, pp. 293--331, 2012, announced at
  SODA'04 and ALENEX'10.

\bibitem{Flajolet85counting}
P.~Flajolet and G.~N. Martin, ``Probabilistic counting algorithms for data base
  applications,'' \emph{Journal of Computer and System Sciences}, vol.~31,
  no.~2, pp. 182--209, 1985, announced at FOCS'83.

\bibitem{Flajolet07hyperloglog}
P.~Flajolet, Éric Fusy, O.~Gandouet, and et~al., ``Hyperloglog: The analysis
  of a near-optimal cardinality estimation algorithm,'' in \emph{In Analysis of
  Algorithms (AOFA)}, 2007.

\bibitem{Heule13hyperloglog}
S.~Heule, M.~Nunkesser, and A.~Hall, ``Hyperloglog in practice: Algorithmic
  engineering of a state of the art cardinality estimation algorithm,'' in
  \emph{Proceedings of the EDBT 2013 Conference}, 2013, pp. 683--692.

\bibitem{boldi11hyperanf}
P.~Boldi, M.~Rosa, and S.~Vigna, ``Hyperanf: Approximating the neighbourhood
  function of very large graphs on a budget,'' in \emph{Proc. 20th WWW}.\hskip
  1em plus 0.5em minus 0.4em\relax ACM, 2011, pp. 625--634.

\bibitem{Cohen14ads}
E.~Cohen, ``All-distances sketches, revisited: Hip estimators for massive
  graphs analysis,'' in \emph{Proc. 33rd ACM Symposium on Principles of
  Database Systems}.\hskip 1em plus 0.5em minus 0.4em\relax ACM, 2014, pp.
  88--99.

\bibitem{li12oneperm}
P.~Li, A.~B. Owen, and C.-H. Zhang, ``One permutation hashing,'' in \emph{Proc.
  26thAdvances in Neural Information Processing Systems}, 2012, pp. 3122--3130.

\bibitem{Shrivastava14oneperm}
A.~Shrivastava and P.~Li, ``Densifying one permutation hashing via rotation for
  fast near neighbor search,'' in \emph{Proc. 31th International Conference on
  Machine Learning (ICML)}, 2014, pp. 557--565.

\bibitem{Shrivastava14densify}
------, ``Improved densification of one permutation hashing,'' in
  \emph{Proceedings of the Thirtieth Conference on Uncertainty in Artificial
  Intelligence, {UAI} 2014, Quebec City, Quebec, Canada, July 23-27, 2014},
  2014, pp. 732--741.

\bibitem{broder97minwise}
A.~Z. Broder, S.~C. Glassman, M.~S. Manasse, and G.~Zweig, ``Syntactic
  clustering of the web,'' \emph{Computer Networks}, vol.~29, pp. 1157--1166,
  1997.

\bibitem{broder98minwise}
A.~Z. Broder, M.~Charikar, A.~M. Frieze, and M.~Mitzenmacher, ``Min-wise
  independent permutations,'' \emph{Journal of Computer and System Sciences},
  vol.~60, no.~3, pp. 630--659, 2000, see also STOC'98.

\bibitem{broder97onthe}
A.~Z. Broder, ``On the resemblance and containment of documents,'' in
  \emph{Proc. Compression and Complexity of Sequences (SEQUENCES)}, 1997, pp.
  21--29.

\bibitem{Charikar02countsketch}
M.~Charikar, K.~Chen, and M.~Farach-Colton, ``Finding frequent items in data
  streams,'' in \emph{Proc. 29th International Colloquium on Automata,
  Languages and Programming (ICALP)}.\hskip 1em plus 0.5em minus 0.4em\relax
  Springer-Verlag, 2002, pp. 693--703.

\bibitem{Li10bbit}
P.~Li, A.~C. K{\"{o}}nig, and W.~Gui, ``b-bit minwise hashing for estimating
  three-way similarities,'' in \emph{Proc. 24thAdvances in Neural Information
  Processing Systems}, 2010, pp. 1387--1395.

\bibitem{Li10bbit2}
P.~Li and A.~C. K{\"{o}}nig, ``b-bit minwise hashing,'' in \emph{Proc. 19th
  WWW}, 2010, pp. 671--680.

\bibitem{li11minhash}
P.~Li, A.~Shrivastava, J.~L. Moore, and A.~C. K{\"o}nig, ``Hashing algorithms
  for large-scale learning,'' in \emph{Proc. 25thAdvances in Neural Information
  Processing Systems}, 2011, pp. 2672--2680.

\bibitem{Bachrach13sketching}
Y.~Bachrach and E.~Porat, ``Sketching for big data recommender systems using
  fast pseudo-random fingerprints,'' in \emph{Proc. 40th International
  Colloquium on Automata, Languages and Programming (ICALP)}, 2013, pp.
  459--471.

\bibitem{bar-yossef02distinct}
Z.~Bar{-}Yossef, T.~S. Jayram, R.~Kumar, D.~Sivakumar, and L.~Trevisan,
  ``Counting distinct elements in a data stream,'' in \emph{Proc. 6th
  International Workshop on Randomization and Approximation Techniques
  (RANDOM)}, 2002, pp. 1--10.

\bibitem{thorup13bottomk}
M.~Thorup, ``Bottom-k and priority sampling, set similarity and subset sums
  with minimal independence,'' in \emph{Proc. 45th ACM Symposium on Theory of
  Computing (STOC)}, 2013.

\bibitem{indyk98ann}
P.~Indyk and R.~Motwani, ``Approximate nearest neighbors: Towards removing the
  curse of dimensionality,'' in \emph{Proc. 30th ACM Symposium on Theory of
  Computing (STOC)}, 1998, pp. 604--613.

\bibitem{andoni08ann}
A.~Andoni and P.~Indyk, ``Near-optimal hashing algorithms for approximate
  nearest neighbor in high dimensions,'' \emph{Communications of the ACM},
  vol.~51, no.~1, pp. 117--122, 2008, see also FOCS'06.

\bibitem{Andoni14lsh}
A.~Andoni, P.~Indyk, H.~L. Nguyen, and I.~Razenshteyn, ``Beyond
  locality-sensitive hashing,'' in \emph{Proc. 25th ACM/SIAM Symposium on
  Discrete Algorithms (SODA)}, 2014, pp. 1018--1028.

\bibitem{motwani95book}
R.~Motwani and P.~Raghavan, \emph{Randomized algorithms}.\hskip 1em plus 0.5em
  minus 0.4em\relax Cambridge University Press, 1995.

\bibitem{mitzenmacher05book}
M.~Mitzenmacher and E.~Upfal, \emph{Probability and computing - randomized
  algorithms and probabilistic analysis}.\hskip 1em plus 0.5em minus
  0.4em\relax Cambridge University Press, 2005.

\bibitem{wegman81kwise}
M.~N. Wegman and L.~Carter, ``New classes and applications of hash functions,''
  \emph{Journal of Computer and System Sciences}, vol.~22, no.~3, pp. 265--279,
  1981, see also FOCS'79.

\bibitem{CRSW11}
L.~E. Celis, O.~Reingold, G.~Segev, and U.~Wieder, ``Balls and bins: Smaller
  hash families and faster evaluation,'' in \emph{Proc. 52nd IEEE Symposium on
  Foundations of Computer Science (FOCS)}, 2011, pp. 599--608.

\bibitem{patrascu12charhash}
M.~P{\v a}tra{\c s}cu and M.~Thorup, ``The power of simple tabulation-based
  hashing,'' \emph{Journal of the ACM}, vol.~59, no.~3, p. Article 14, 2012,
  announced at STOC'11.

\bibitem{indyk01minwise}
P.~Indyk, ``A small approximately min-wise independent family of hash
  functions,'' \emph{Journal of Algorithms}, vol.~38, no.~1, pp. 84--90, 2001,
  see also SODA'99.

\bibitem{patrascu10kwise-lb}
M.~P{\v a}tra{\c s}cu and M.~Thorup, ``On the $k$-independence required by
  linear probing and minwise independence,'' in \emph{Proc. 37th International
  Colloquium on Automata, Languages and Programming (ICALP)}, 2010, pp.
  715--726.

\bibitem{PP08}
A.~Pagh and R.~Pagh, ``Uniform hashing in constant time and optimal space,''
  \emph{SIAM J. Comput.}, vol.~38, no.~1, pp. 85--96, 2008.

\bibitem{Christiani15indep}
T.~Christiani, R.~Pagh, and M.~Thorup, ``From independence to expansion and
  back again,'' 2015, to appear.

\bibitem{carter77universal}
L.~Carter and M.~N. Wegman, ``Universal classes of hash functions,''
  \emph{Journal of Computer and System Sciences}, vol.~18, no.~2, pp. 143--154,
  1979, see also STOC'77.

\bibitem{thorup13doubletab}
M.~Thorup, ``Simple tabulation, fast expanders, double tabulation, and high
  independence,'' in \emph{FOCS}, 2013, pp. 90--99.

\bibitem{siegel04hash}
A.~Siegel, ``On universal classes of extremely random constant-time hash
  functions,'' \emph{SIAM Journal on Computing}, vol.~33, no.~3, pp. 505--543,
  2004, see also FOCS'89.

\bibitem{zobrist70hashing}
A.~L. Zobrist, ``A new hashing method with application for game playing,''
  Computer Sciences Department, University of Wisconsin, Madison, Wisconsin,
  Tech. Rep.~88, 1970.

\bibitem{PT13:twist}
M.~P\v{a}tra\c{s}cu and M.~Thorup, ``Twisted tabulation hashing,'' in
  \emph{Proc. 24th ACM/SIAM Symposium on Discrete Algorithms (SODA)}, 2013, pp.
  209--228.

\bibitem{serfling74replacement}
R.~J. Serfling, ``Probability inequalities for the sum in sampling without
  replacement,'' \emph{Annals of Statistics}, vol.~2, no.~1, pp. 39--48, 1974.

\bibitem{dietzfel03tabhash}
M.~Dietzfelbinger and P.~Woelfel, ``Almost random graphs with simple hash
  functions,'' in \emph{Proc. 25th ACM Symposium on Theory of Computing
  (STOC)}, 2003, pp. 629--638.

\bibitem{Goodrich11ibt}
M.~T. Goodrich and M.~Mitzenmacher, ``Invertible bloom lookup tables,'' in
  \emph{2011 49th Annual Allerton Conference on Communication, Control, and
  Computing, Allerton Park {\&} Retreat Center, Monticello, IL, USA, 28-30
  September, 2011}, 2011, pp. 792--799.

\bibitem{eppstein2011s}
D.~Eppstein, M.~T. Goodrich, F.~Uyeda, and G.~Varghese, ``What's the
  difference?: efficient set reconciliation without prior context,'' in
  \emph{ACM SIGCOMM Computer Communication Review}, vol.~41, no.~4.\hskip 1em
  plus 0.5em minus 0.4em\relax ACM, 2011, pp. 218--229.

\bibitem{eppstein2011straggler}
D.~Eppstein and M.~T. Goodrich, ``Straggler identification in round-trip data
  streams via newton's identities and invertible bloom filters,''
  \emph{Knowledge and Data Engineering, IEEE Transactions on}, vol.~23, no.~2,
  pp. 297--306, 2011.

\bibitem{mitzenmacher2012biff}
M.~Mitzenmacher and G.~Varghese, ``Biff (bloom filter) codes: Fast error
  correction for large data sets,'' in \emph{Information Theory Proceedings
  (ISIT), 2012 IEEE International Symposium on}.\hskip 1em plus 0.5em minus
  0.4em\relax IEEE, 2012, pp. 483--487.

\end{thebibliography}
